\documentclass{amsart}
\pdfoutput=1
\usepackage[english]{babel}
\usepackage{graphicx}
\usepackage[square, numbers]{natbib}
\usepackage{amsmath, amssymb, amsthm, subfig}
\usepackage[amssymb]{SIunits}

\usepackage[colorinlistoftodos]{todonotes}

\usepackage{tikz}
\usetikzlibrary{arrows,fit, shapes,positioning}

\usepackage{xr}
\usepackage{paralist}

\newcommand{\parttens}{\otimes_{\boldsymbol{m},\boldsymbol{n}}}

\newcommand{\fullsol}{p_{\boldsymbol{m}\cdot\boldsymbol{n}}(t)}
\newcommand{\modessol}{v_{\boldsymbol{m}}(t)}
\newcommand{\statQsol}{\statQ_{\boldsymbol{n}}}

\newcommand{\ca}{Ca$^{2+}$}
\newcommand{\ipthree}{IP$_{3}$}
\newcommand{\ipr}{\ipthree{R}}
\newcommand{\iprOne}{type~I~\ipr}
\newcommand{\iprTwo}{type~II~\ipr}

\newcommand{\Po}{P_O}

\usepackage{booktabs}
\usepackage{color}

\newcommand{\prob}{\mathbb{P}}

\DeclareMathOperator{\diag}{diag}
\DeclareMathOperator{\id}{id}

\newcommand{\fullM}{M}
\newcommand{\modesM}{\tilde{\fullM}}
\newcommand{\fullEv}{\zeta}
\newcommand{\modesEv}{\tilde{\fullEv}}

\newcommand{\Mi}{\text{M}^i}

\newcommand{\seqM}{S^k}
\newcommand{\seqOC}{T^k}

\newcommand{\modesMij}{\modesM^i_j}

\newcommand{\nM}{n_M}

\newcommand{\Qik}{Q^i_k}

\newcommand{\Qi}{Q^i}
\newcommand{\Qone}{\textcolor{parkcolor}{$Q^1$}}
\newcommand{\Qtwo}{\textcolor{drivecolor}{$Q^2$}}

\newcommand{\statQ}{\pi}
\newcommand{\statQi}{\statQ^i}

\newcommand{\statFullM}{\mu}
\newcommand{\statM}{\tilde{\statFullM}}
\newcommand{\statMi}{\statM^i}

\newcommand{\pnull}{p^0_{\boldsymbol{m} \cdot \boldsymbol{n}}}

\newcommand{\pMnull}{\tilde{m}_0}

\newcommand{\pQinull}{p^i}

\newcommand{\hmm}{\left( (\pMnull, \modesM), (\pQinull,\Qi)_{i=1}^{\nM} \right)}

\colorlet{darkgreen}{green!40!black}
\colorlet{parkcolor}{blue}
\colorlet{drivecolor}{brown!70!black}
\newcommand{\Mone}{\textcolor{parkcolor}{M$^1$}}
\newcommand{\Mtwo}{\textcolor{drivecolor}{M$^2$}}
\newcommand{\gO}{\textcolor{darkgreen}{O}}
\newcommand{\rC}{\textcolor{red}{C}}

\newtheorem{proposition}{\bf Proposition}[section]
\newtheorem{lemma}{Lemma}[section]
\newtheorem{definition}{Definition}[section]

\newtheorem{remark}{\bf Remark}[section]
\usepackage[colorlinks]{hyperref}
\usepackage{fullpage}
\usepackage{mathpazo}
\makeatletter
\newcommand{\addresseshere}{%
  \enddoc@text\let\enddoc@text\relax
}
\makeatother

\begin{document}

\title{Modelling modal gating of ion channels with hierarchical Markov models}

\author[Ivo Siekmann,  Mark Fackrell, Edmund J. Crampin and Peter Taylor]{Ivo Siekmann$^{1,2}$,  Mark Fackrell$^{4}$, Edmund J. Crampin$^{1,2,3,4,5}$ and Peter Taylor$^{4}$}


\address{
  \begin{minipage}[t]{1.0\linewidth}
$^{1}$
\begin{minipage}[t]{1.0\linewidth}
Systems Biology Laboratory, Melbourne School of Engineering, University of Melbourne , Australia
\end{minipage}
\\
  $^{2}$
  \begin{minipage}[t]{1.0\linewidth}
    Centre for Systems Genomics, University of Melbourne, Australia
  \end{minipage}
  \\
  $^{3}$
  \begin{minipage}[t]{1.0\linewidth}
    ARC Centre of Excellence in Convergent Bio-Nano Science and
    Technology, 
    Australia
  \end{minipage}
\\
  $^{4}$
  \begin{minipage}[t]{1.0\linewidth}
    School of Mathematics and Statistics, University of Melbourne, Australia
  \end{minipage}
\\
  $^{5}$
  \begin{minipage}[t]{1.0\linewidth}
    School of Medicine, University of Melbourne, Australia
  \end{minipage}
\end{minipage}
}

\maketitle
\noindent\addresseshere 
%
\begin{abstract}
  Many ion channels spontaneously switch between different levels of
  activity. Although this behaviour known as modal gating has been
  observed for a long time it is currently not
  well understood. Despite the fact that appropriately representing
  activity changes is essential for accurately capturing time course
  data from ion channels
  , systematic approaches for 
  modelling modal gating are currently not available. In this paper,
  we develop a modular approach for building such a
  model 
  in an iterative process. First, stochastic switching between modes
  and stochastic opening and closing within modes are represented in
  separate aggregated Markov models. 
  Second, the continuous-time hierarchical Markov model, a new
  modelling framework proposed here
  , then enables us 
  to combine these components
  so that in the integrated model both mode switching as well as the
  kinetics within modes are appropriately represented. A mathematical analysis
  reveals that 
  the behaviour of the hierarchical Markov model naturally depends on
  the properties of its components.  We also demonstrate how a
  hierarchical Markov model can be parameterised using experimental
  data and show that it provides a better representation than a
  previous model of the same data set. Because evidence is increasing
  that modal gating reflects underlying molecular properties of the
  channel protein, it is likely that biophysical processes are better
  captured by our new approach than in earlier models.
\end{abstract}


\section{Introduction}
\label{sec:intro}

Ion channels regulate the flow of ions across the cell membrane by
stochastic opening and closing. As soon as it became possible to
detect currents generated by the movement of charged ions through the
channel via the patch-clamp technique \citep{Neh:76a}, \citet{Col:81a}
developed the theory of modelling single ion channels with
continuous-time Markov models which describe the time-course of
opening and closing that is reflected in single-channel currents by
stochastic jumps between zero (closed) and one or more small non-zero
current levels in the~pA range (open). 
The activity of an ion channel is usually measured by its open
probability~$\Po$. But by 1983, \citet{Mag:83a, Mag:83b} had already
observed spontaneous changes between different levels of channel
activity in the calcium-activated potassium channel. Since then this
phenomenon, known as modal gating, has been ubiquitously observed
across a wide range of ion channels but the significance of modal
gating has remained unclear.

In this study we present a general framework for building data-driven
models of ion channels that account for modal gating. 
This is essential for accurately representing the dynamics of an ion
channel---instead of producing a misleading constant intermediate open
probability~$\Po$, a model should represent the switching between
highly different levels of activity characteristic of each mode. {This
  is illustrated in Figure~\ref{fig:changepoint} where data points
  labelled~\Mone\ form a segment characterised by a low open
  probability whereas, the segment labelled~\Mtwo\ is characterised by
  a high open probability. In a realistic time series, the changes
  between~\Mone\ and~\Mtwo\ occur on a time scale so slow that a model
  fitted directly to the sequence of closed and open events would not
  be able to resolve this. Thus, instead of infrequent switching
  between high and low open probabilities, a model fitted directly to
  the data would show an intermediate open probability rather than
  switching between high and low open probabilities.} On the other
hand, modes of an ion channel have been associated with biophysical
properties of the channel protein \citep{Sie:14a}. Therefore, a model
accounting for modal gating is more likely to appropriately relate the
dynamics of ion channels to underlying biophysical states of the
channel protein.

{Nevertheless, except for two recent models of the inositol
  trisphosphate receptor~(\ipr), see \citet{Ull:12a, Sie:12a}, modal
  gating is usually not considered in ion channel models.}  One
difficulty in appropriately representing modal gating of ion channels
in a model is the fact that for a time series of measurements
collected from an ion channel it is impossible to infer directly in
which mode the channel is at a given point in time. However,
\citet{Sie:14a} have shown how this information can be obtained by
statistical changepoint analysis, see
Figure~\ref{fig:changepoint}. {The method identifies significant
  changes of the open probability between adjacent segments in time
  series of open and closed events recorded from an ion channel}.

\begin{figure}[htbp]
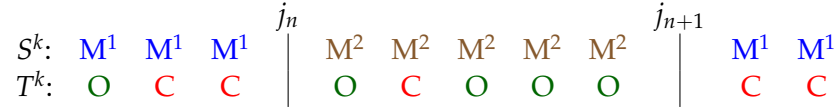

      \centering
      {\large
        \begin{tabular}{*{13}{c}}
          & & & & $j_n$ &&&&&& $j_{n+1}$\\
         $\seqM$: & \Mone & \Mone & \Mone & \vline & \Mtwo & \Mtwo & \Mtwo & \Mtwo & \Mtwo & \vline& \Mone & \Mone \\
          $\seqOC$: & \gO & \rC & \rC & {\vline} & \gO & \rC & \gO & \gO & \gO & {\vline} & \rC & \rC
        \end{tabular}
      }%
      \caption{After a statistical analysis of modal gating
        \citep{Sie:14a}, changepoints~$j_n$ 
        have been inferred for a time series of ion channel
        data. Through this segmentation, the original time
        series~$\seqOC$ of open~(\gO) and closed~(\rC) events has been
        augmented by the additional information~$\seqM$ of the
        mode~(\Mone, \Mtwo \dots) that the channel is in for a given
        point in time.}
  \label{fig:changepoint}
\end{figure}

As a result, after this analysis has been carried out, for each point
in the time series it is not only known if the channel is open~(\gO)
or closed~(\rC) but also, in which of the modes~\Mone, \Mtwo, \dots\
the channel is. {Previously, we observed stochastic
  switching between a nearly inactive mode~\Mone\ and a highly active
  mode~\Mtwo\ in data from the~\ipr\ \citep{Sie:14a}. In this paper we
  will represent the stochastic process of switching between an
  arbitrary number of different modes~$\Mi$ by a continuous-time
  Markov model with infinitesimal generator~$\modesM$. For data by
  \citet{Wag:12a}, empirical histograms suggest that the sojourn time
  distribution~$f_{\text{\Mone}}(t)$ within mode~\Mone\ is not
  exponential (see Figures~5 and 6 in \citet{Sie:14a} and
  Figure~\ref{fig:IPR1Ca10nMPark}
  ). For this reason, in general, more than one state is needed for
  accurately representing the process of switching between modes. This
  means that modal sojourn times are represented by phase-type
  distributions, a class of distributions which is defined by the time
  a Markov chain spends in a set of transient states until exiting to
  an absorbing state \citep{Neu:75a, 
    Neu:81a}. 
  We assume that the infinitesimal generator~$\modesM$ representing
  the switching between modes~$\Mi$, $i=1, \dots\ \nM$, has the
  following block structure:

\begin{equation}
  \label{eq:Mmodel}
  \modesM = \begin{pmatrix}
    \modesM^{1,1} & | & \modesM^{1,2} & | & \dots & | & \modesM^{1,\nM}\\
    \hline
    \modesM^{2,1} & | & \modesM^{2,2} & |& \dots & | & \modesM^{2,\nM}\\
    \vdots &  & & \ddots& & &\vdots\\
\vdots & & & & \ddots & & \vdots\\
    \vdots & & & & &\ddots & \vdots\\
    \modesM^{\nM,1} & | & \dots  & \dots &\dots & | & \modesM^{\nM,\nM}
  \end{pmatrix},
\end{equation}
where the block
matrices~$\modesM^{i,i} \in \mathbb{R}^{m_i\times m_i}$,
$m_i \in \mathbb{N}$, on the diagonal describe transitions between
states that represent the same mode~$\Mi$ whereas the off-diagonal
blocks~$\modesM^{i, j} \in \mathbb{R}^{m_i \times m_j}$ represent
transitions between states representing different modes~$\Mi$
and~$\text{M}^j$, $i\neq j$. An example for a model for switching
between two modes~\Mone\ and \Mtwo\ is shown in
Figure~\ref{fig:switchingmodel}.

Our modal gating analysis illustrated in Figure~\ref{fig:changepoint}
not only enables us to represent the stochastic process of switching
\emph{between} modes~$\Mi$ but by studying the dynamics within
representative segments we can investigate the processes of stochastic
opening and closing characteristic of each mode. For the example in
Figure~\ref{fig:changepoint} the dynamics \emph{within} mode~\Mtwo\
can be analysed by considering the sequence of open and closed events
between~$j_k$ and $j_{k+1}$. The dynamics within a mode~$\Mi$ can be
represented by a Markov model with infinitesimal generator~$\Qi$ which
is obtained by fitting to representative segments of the same mode
\citep{Sie:12a}. Similar to the sojourn times in the modes~$\Mi$, the
open and closed time distributions~$f_O(t)$ and~$f_C(t)$,
respectively, are non-exponential and more than one open or closed
state may be needed for accurately representing the dynamics. For the
example shown in Figure~\ref{fig:changepoint} we obtain two models
with infinitesimal generators~\Qone\ and~\Qtwo, see
Figure~\ref{fig:kineticsmodel}.


In this paper we develop a new mathematical model, the continuous-time
hierarchical Markov model, that accounts simultaneously for both
transitions \emph{between} modes as well as the stochastic opening and
closing \emph{within} modes.  Whereas a hierarchical Markov model in
discrete time has been previously described \cite{Fin:98a} we are not
aware of a continuous-time version discussed in the literature, so we
develop the mathematical theory in detail and prove some fundamental
properties. For the example of modal gating we assume that switching
between modes~$\Mi$ is a top-level process that regulates the
bottom-level process, the opening and closing of the channel
characteristic of a particular mode~$\Mi$. This is illustrated in
Figure~\ref{fig:modalGating}. 

\colorlet{darkgreen}{green!40!black}
  \tikzset{StayStyle/.style = {shape          = circle,
      ball color     = green!25!white,
      text           = darkgreen,
      minimum size   = 24 pt}}

  \tikzset{StayBlackStyle/.style = {shape          = circle,
      ball color     = green!25!white,
      text           = black,
      minimum size   = 24 pt}}

  \tikzset{LeaveStyle/.style = {shape          = circle,
      ball color     = red!20!white,
      text           = red,
      minimum size   = 24 pt}}

  \tikzset{ModeOneStyle/.style = {shape          = circle,
      ball color     = gray!70,
      text           = parkcolor,
      minimum size   = 24 pt}}

  \tikzset{ModeTwoStyle/.style = {shape          = circle,
      ball color     = orange!80!white,
      text           = brown!50!black,
      minimum size   = 24 pt}}

\tikzset{EdgeStyle/.style   = {
    very thick,
}}
\begin{figure}[htbp]
\subfloat[inter-modal transitions]{%
\label{fig:switchingmodel}
  \raggedright

  \includegraphics[width=0.8\textwidth]{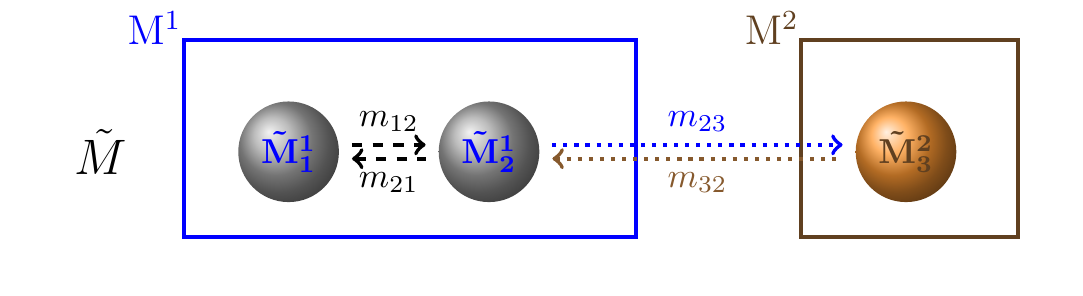}%
}\\

\subfloat[intra-modal dynamics]{
  \label{fig:kineticsmodel}
  \raggedleft

\qquad \qquad \qquad 
\includegraphics[width=0.58\textwidth]{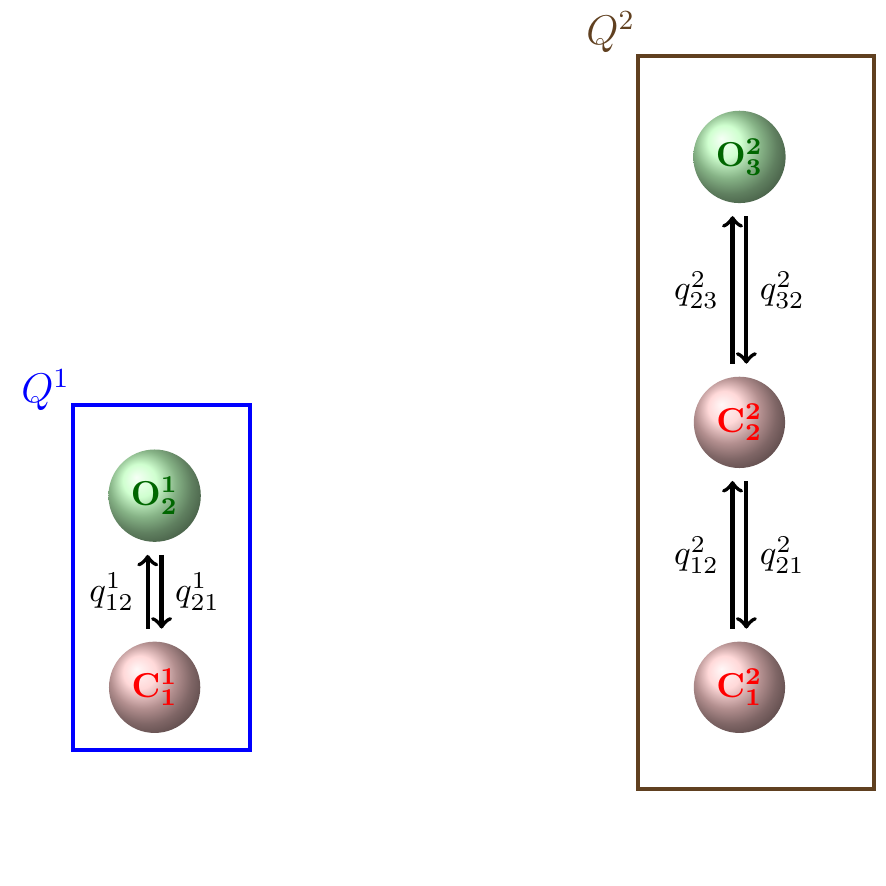}%
}
\caption{Modular components of a model for modal gating. (a) gives an
  example for an aggregated Markov model~$\modesM$ representing
  inter-modal dynamics, the stochastic switching between two modes,
  \Mone\ and \Mtwo. \Mone\ is modelled by an aggregate of two states
  whereas~\Mtwo\ is represented by one state. The
  rates~\textcolor{parkcolor}{$m_{23}$} and
  \textcolor{drivecolor}{$m_{32}$} stand for transitions between both
  modes. Note that~$\modesM$ may in general represent transitions
  between more than two modes, therefore the states~$\modesM^i_j$ are
  numbered consecutively by subscripts~$j$ whereas the
  superscripts~$i$ indicate the mode~$\Mi$. (b) shows models~\Qone\
  and \Qtwo\ representing the stochastic opening and closing that is
  characteristic of mode~\Mone\ or \Mtwo, respectively. The
  states~{$C^i_k$} and~\textcolor{darkgreen}{$O^i_k$} are numbered
  similarly to the~$\modesM^{i}_{j}$. Note that~$k=1, \dots, n_i$ for
  each mode~$\Mi$ in contrast to the states~$\modesM^{i}_j$ where the
  index~$j$ runs from 1 to the total number of states. In
  Figure~\ref{fig:modalamm} we show how~$\modesM$ and
  the~$\Qi$\text{s} are combined in a model that accurately represents
  both inter-modal transitions as well as intra-modal kinetics.}
  \label{fig:modalGating}
\end{figure}

The states~$\modesM^i_j$ are numbered consecutively by subscripts~$j$
whereas the superscripts~$i$ indicate the mode~$\Mi$.  While the model
is in mode~\Mone or analogously within one of the
states~\textcolor{parkcolor}{$\modesM^1_1$}
or~\textcolor{parkcolor}{$\modesM^1_2$}
(Figure~\ref{fig:switchingmodel}), its opening and closing is
described by the infinitesimal generator~\Qone\
(Figure~\ref{fig:kineticsmodel}). As soon as~\Mone\ is left to
state~\textcolor{drivecolor}{$\modesM^2_3$}, the current state of
model~\Qone\ is vacated and a state of model~\Qtwo\ is entered. Now,
opening and closing is accounted for by~\Qtwo\ until the
state~\textcolor{drivecolor}{$\modesM^2_3$} and mode~\Mtwo\ is left
and state~\textcolor{parkcolor}{$\modesM^1_2$} is entered.

The transitions between modes described via~$\modesM$ and the dynamics
within modes captured by~$\Qi$ illustrated in
Figure~\ref{fig:modalGating} can be represented in a Markov model with
infinitesimal generator~$\fullM$ that is derived from the individual
components~$\modesM$ and~$\Qi$. The idea is illustrated in
Figure~\ref{fig:modalamm} and developed formally in
Section~\ref{sec:methods}.

\begin{figure}
  \centering
\includegraphics[width=0.9\textwidth]{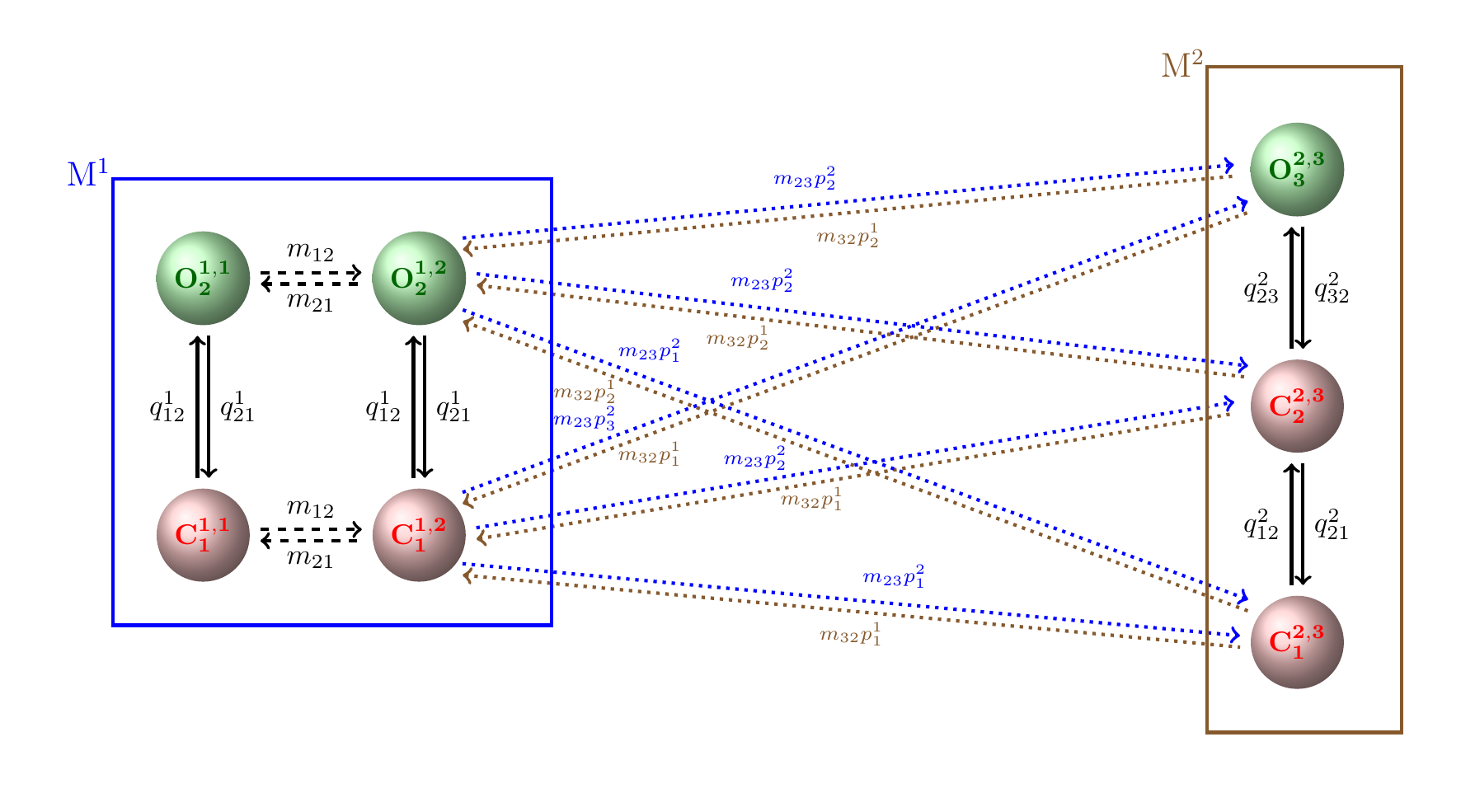}%
\caption{Aggregated Markov model that represents both transitions
  between modes~\Mone\ and \Mtwo\ according to model~$\modesM$
  (Figure~\ref{fig:switchingmodel}) as well as stochastic opening and
  closing consistent with models~\Qone\ and \Qtwo\
  (Figure~\ref{fig:kineticsmodel}). The open and closed states
  are~\textcolor{darkgreen}{$O^{i,j}_k$}
  and~\textcolor{red}{$C^{i,j}_k$}, respectively, where the
  superscripts~$i$,$j$ refer to the state~$\modesM^{i}_{j}$ in the
  model shown in Figure~\ref{fig:switchingmodel} whereas the
  subscript~$k$ is the index of the state within a model~$\Qi$ shown
  in Figure~\ref{fig:kineticsmodel}. This illustrates that the state
  set of the full model is obtained by the Cartesian product of states
  representing the modes~$\Mi$ with the states of the model~$\Qi$. Due
  to the transitions~$m_{12}$ and~$m_{21}$ between the two states
  representing~\Mone, in the full model there are two copies of
  model~\Qone\ connected by transition rates~$m_{12}$
  and~$m_{21}$. For transitions between modes, the
  rates~\textcolor{parkcolor}{$m_{23}$} exiting~\Mone\
  and~\textcolor{drivecolor}{$m_{32}$} exiting~\Mtwo\ are weighted
  with stochastic vectors~$p^1=(p^1_1,p^1_2)$
  and~$p^2=(p^2_1,p^2_2, p^2_3)$ that can be interpreted as initial
  distributions when entering~\Mone\ or \Mtwo.  }
  \label{fig:modalamm}
\end{figure}

In order to account for the states~$\modesM^{i}_j$
as well as the states~\textcolor{darkgreen}{$O^{i}_k$}
and~\textcolor{red}{$C^{i}_k$} representing the opening and closing
within~$\Mi$, the state space of the full model consists of the
Cartesian products of the $\modesM^{i}_j$ with
the~\textcolor{darkgreen}{$O^{i}_k$}
and~\textcolor{red}{$C^{i}_k$}. Thus, the state space of the full
model consists of open and closed
states~\textcolor{darkgreen}{$O^{i,j}_k$}
and~\textcolor{red}{$C^{i,j}_k$}, respectively, where the
superscripts~$i$,$j$ refer to the state~$\modesM^{i}_{j}$ in the model
shown in Figure~\ref{fig:switchingmodel} whereas the subscript~$k$ is
the index of the state within a model~$\Qi$ shown in
Figure~\ref{fig:kineticsmodel}. For the example shown in the figure,
the closed states~\textcolor{red}{$C^{1,1}_1$}
and~\textcolor{red}{$C^{1,2}_1$} as well as the open
states~\textcolor{darkgreen}{$O^{1,1}_2$}
and~\textcolor{darkgreen}{$O^{1,2}_2$} are connected by the transition
rates~$m_{12}$ and~$m_{21}$. Because~\Mone\ is modelled by two
states~$\modesM^{1}_1$ and~$\modesM^{1}_2$, two ``copies'' of~\Qone\
appear in the full model whereas there is only one ``copy'' of~\Qtwo\
which is represented by only one state in~$\modesM$. For transitions
between modes, the rates~\textcolor{parkcolor}{$m_{23}$}
exiting~\Mone\ and~\textcolor{drivecolor}{$m_{32}$} exiting~\Mtwo\ are
weighted with stochastic vectors~$p^1=(p^1_1,p^1_2)$
and~$p^2=(p^2_1,p^2_2, p^2_3)$ that can be interpreted as initial
distributions when entering~\Mone\ or
\Mtwo. 
The mathematical details of the construction of this model are
presented in Section~\ref{sec:methods}. }

It is a strength of our approach that it enables us to build
data-driven models of modal gating in a modular way.  After segmenting
ion channel data with the method by \citet{Sie:14a} we obtain a
stochastic sequence of events~$\Mi$ that describes the time course of
transitions between different modes. The infinitesimal
generators~$\modesM$ and the~$\Qi$ can then be parameterised from
these data. We demonstrate the practical implementation of this
approach in Section~\ref{sec:results} using experimental data by
\citet{Wag:12a} and compare the results with our previously published
model of the same data set \citep{Sie:12a}.

{
  We investigate the mathematical structure of the continuous-time
  hierarchical Markov model in more detail in
  Section~\ref{sec:analysis}. In particular we show that many
  important properties of the infinitesimal generator~$\fullM$ of the
  full model can be derived from the generators~$\modesM$
  and~$\Qi$. We expect that similar to its discrete-time counterpart
  \citep{Fin:98a}, the continuous-time hierarchical Markov model will
  have a variety of applications beyond the modelling of modal gating
  considered here.  } 

We discuss our approach to modal gating in
Section~\ref{sec:conclusions}. In particular we explain why our new
modelling framework is not only a better representation of ion channel
dynamics but also more likely than other modelling approaches to
provide a structure that realistically captures biophysical processes.

\newpage

\section{Methods}
\label{sec:methods}

\subsection{Preliminaries}
\label{sec:prelim}

We now develop formally the hierarchical Markov model illustrated
graphically in Figures~\ref{fig:modalGating}
and~\ref{fig:modalamm}. {First, let us describe the structure of the
  probability distribution~$p$ over the states of the hierarchical
  Markov model.}  Let~$v=(v^1; v^2; \dots; v^{\nM})$ denote a state
probability distribution 
of the model~$\modesM$. That is, for~$i=1, \dots, \nM$, $v^i$ is the
probability distribution of the states in mode~$\Mi$
. In general, we will allow~$\modesM$ to be an aggregated
Markov model so that each of the components~$v^i$ of the
vector~$v$ may itself be a vector. We make the convention that
components~$v^i$ and~$v^j$ that are meant to refer to a vector
are separated by semicolons, whereas components of a vector are
separated by commas. Let us first assume for simplicity that all
modes~$\Mi$ are represented by only one state so that the
components~$v^i$ are scalars. Then the distribution~$p$ over the
states of the full model~$\fullM$ is a weighting of the
distributions~$w^i$ over the distributions over the states of the
models~$\Qi$. Thus, we
obtain~$p:=(v^1 \cdot w^1; \dots; v^i\cdot w^i; \dots;
v^{\nM} \cdot w^{\nM})$.
Here `$\cdot$' denotes scalar multiplication of vectors~$w^i$ with
scalars~$v^i$. If more than one state is needed for representing
the modes~$\Mi$ 
we must generalise appropriately the ``weighting'' of a
vector~$w^i$ with a vector~$v^i$.  Such a generalisation is
provided by the tensor product~`$\otimes$'.

\begin{definition}[Kronecker product $\otimes$]
  We will only need the special case of the tensor product for
  matrices, the Kronecker product. Let $A \in \mathbb{R}^{m\times n}$,
  $B \in \mathbb{R}^{p \times r}$. Then

  \begin{equation}
    \label{eq:kroneckerprod}
    A \otimes B := (a_{ij} \cdot B)_{1 \leq i \leq m, 1 \leq j \leq n}
    =
    \begin{pmatrix}
      a_{11} B& \dots & a_{1n} B\\
      \vdots & \ddots & \vdots\\
      a_{m1} B & \dots & a_{mn} B
    \end{pmatrix} \in \mathbb{R}^{mp \times nr}.
  \end{equation}
  The Kronecker product also applies to vectors by identifying column
  vectors with $(m \times 1)$- and row vectors
  with~$(1 \times m)$-matrices.
\end{definition}




\begin{definition}[Kronecker sum $\oplus$]
The Kronecker sum of square matrices~$A \in \mathbb{R}^{m \times m}$
and~$B \in \mathbb{R}^{n \times n}$ is

\begin{equation}
  \label{eq:kroneckersum}
  A \oplus B := A \otimes \id_n + \id_m \otimes B \in \mathbb{R}^{mn
    \times mn}, 
\end{equation}
where~$\id_m$ and~$\id_n$ are the identity matrices of the respective
dimensions
.
\end{definition}

For some properties of Kronecker product and sum that we require for
our analysis of the hierarchical Markov model
(Section~\ref{sec:analysis}) we refer to
Appendix~\ref{sec:mathsbackground}. For a distribution~$v$ over the
states of an aggregated Markov model, subvectors that represent the
distributions over the states of the same mode~$\Mi$ can be naturally
described by partitions.

\begin{definition}[Partitioned vectors, multi-indices
  ]
  \label{def:partvector}
  A multi-index is any
  vector~$\boldsymbol{\alpha}=(\alpha_1, \dots, \alpha_d) \in
  \mathbb{N}^d$.  We define the absolute
  value~$|\boldsymbol{\alpha}|=\sum_{i=1}^d \alpha_i$ and
  denote~$\dim(\boldsymbol{\alpha})=d$ the dimension of~$\boldsymbol{\alpha}$.\\
  A vector~$v
  $ is
  partitioned by a multi-index~$\boldsymbol{\alpha}$ if
  \[
  v_{\boldsymbol{\alpha}}:=(v^1; \dots; v^i; \dots; v^{\dim (\boldsymbol{\alpha})})
\]
and for each~$i$ we have~$v^i \in \mathbb{R}^{\alpha_i}$. \\
Selection of the~$i$-th partition of~$v_{\boldsymbol{\alpha}}$ is written as
\[
v_{\boldsymbol{\alpha}}(i)=v^i.
\]
The vector space of $\boldsymbol{\alpha}$-partitioned vectors~$v_{\boldsymbol{\alpha}}$ is denoted~$\mathbb{R}^{\boldsymbol{\alpha}}$.
\end{definition}

How distributions~$p$ over the states of a hierarchical Markov model
relate to distributions over the states of~$\modesM$ and~$\Qi$ can be
clarified by the tensor product of partitioned vector spaces.




\begin{definition}[Tensor product~$\mathbb{R}^{\boldsymbol{m}} \parttens \mathbb{R}^{\boldsymbol{n}}$ of~$d$-partitioned vector spaces]
\label{def:parttens}
  Let~$\boldsymbol{m}, \boldsymbol{n} \in \mathbb{N}^d$,
  $v_{\boldsymbol{m}} \in \mathbb{R}^{\boldsymbol{m}}$,
  $w_{\boldsymbol{n}} \in \mathbb{R}^{\boldsymbol{n}}$ be
  $d$-partitioned vectors. Then the tensor
  product~$u_{\boldsymbol{m} \cdot \boldsymbol{n}}$ of~$d$-partitioned vectors~$v_{\boldsymbol{m}}$ and~$w_{\boldsymbol{n}}$ is defined by

  \begin{equation}
    \label{eq:tensorpartition}
    u_{\boldsymbol{m} \cdot \boldsymbol{n}}:=v_{\boldsymbol{m}} \parttens w_{\boldsymbol{n}} := (v^1 \otimes w^1; \dots; v^i \otimes w^i; \dots; v^{d} \otimes w^{d}),
  \end{equation}
  with the component-wise product~${\boldsymbol{m}\cdot\boldsymbol{n}}$
  of $\boldsymbol{m}$ and $\boldsymbol{n}$.  With the tensor product
  `$\parttens$' we obtain the vector space

\[
\mathbb{R}^{\boldsymbol{m}} \parttens \mathbb{R}^{\boldsymbol{n}}
\]
of the $d$-partitioned vector spaces~$\mathbb{R}^{\boldsymbol{m}}$ and $\mathbb{R}^{\boldsymbol{n}}$.
\end{definition}

\begin{remark} We make some remarks regarding the interpretation of Definition~\ref{def:parttens}:
  \begin{compactitem}
  \item It can be easily verified that `$\parttens$' fulfils the
    properties of a tensor product on the vector
    space~$\mathbb{R}^{\boldsymbol{m}} \parttens
    \mathbb{R}^{\boldsymbol{n}}$.
  \item
    Vectors~$u_{\boldsymbol{m} \cdot \boldsymbol{n}} \in
    \mathbb{R}^{\boldsymbol{m}} \parttens \mathbb{R}^{\boldsymbol{n}}$
    can be written as linear combinations
\begin{equation}
\label{eq:elementparttens}
u_{\boldsymbol{m} \cdot \boldsymbol{n}}=\sum_{k=1}^{d}\sum_{i=1}^{m_k}\sum_{j=1}^{n_k} a^k_{ij} (v^{k,i}_{\boldsymbol{m}} \parttens w^{k,j}_{\boldsymbol{n}}), \quad a^k_{i,j} \in \mathbb{R}
\end{equation}
where~$d=\dim{\boldsymbol{m}}=\dim{\boldsymbol{n}}$.  By choosing bases~$\{v^{k,i} \}$, $i=1, \dots, m_k$, $\{w^{k,j} \}$, $j=1, \dots, n_k$, we obtain systems of linearly independent vectors
\begin{align*}
v^{k,i}_{\boldsymbol{m}}&=(0; \dots; v^{k,i}; \dots; 0) \in \mathbb{R}^{\boldsymbol{m}}\\
w^{k,j}_{\boldsymbol{n}}&=(0; \dots; w^{k,j}; \dots; 0) \in \mathbb{R}^{\boldsymbol{n}}
\end{align*}
Thus, from~\eqref{eq:elementparttens} it is easy to see that
\[\mathbb{R}^{\boldsymbol{m}} \parttens \mathbb{R}^{\boldsymbol{n}}
\cong \mathbb{R}^{\boldsymbol{m}\cdot\boldsymbol{n}}\]
where~${\boldsymbol{m}\cdot\boldsymbol{n}}$ again denotes the
component-wise product of $\boldsymbol{m}$ and $\boldsymbol{n}$.
\end{compactitem}
\end{remark}

\subsection{A hierarchical Markov model for modal gating}
\label{sec:mgamm}

Based on the block structure~\eqref{eq:Mmodel} of~$\modesM$ we now
show how a transition matrix for the full model can be calculated from
its components~$\hmm$. Let~$\boldsymbol{m}$ and~$\boldsymbol{n}$ be
the multi-indices defined above. The transitions within the
modes~$\Mi$ are represented in the full model by block
matrices~$\fullM^{i,i}=\modesM^{i,i} \oplus \Qi \in \mathbb{R}^{m_i
  n_i \times m_i n_i}$.
It follows that~$\dim \fullM^{i,i} = m_i
n_i$. 
Moreover, we define the matrix of
initial conditions for a transition from~$\Qi$ to~$Q^j$ by

\begin{equation}
  \label{eq:Pij}
  P^{i,j} = u_{n_i}^T \otimes p^j = p^j \otimes u_{n_i} ^T,
\end{equation}
where the row vector~$p^j \in \mathbb{R}^{1 \times n_j}$ is the
initial condition for~$Q^j$ from Definition~\ref{def:components},
and~$u_{n_i}^T \in \mathbb{R}^{n_i \times 1}$ is a column vector of
ones. We observe that $P^{i,j} \in \mathbb{R}^{n_i \times n_j}$ so
that, for $i \neq j$ we
have~$\fullM^{i,j}=\modesM^{i,j} \otimes P^{i,j} \in \mathbb{R}^{m_i
  n_i \times m_j n_j}$
. We can now define the components of a continuous-time hierarchical
Markov model and calculate its infinitesimal generator:


{Analogous to the discrete-time hierarchical Markov
  model by \citet{Fin:98a}, we define a continuous-time hierarchical
  Markov model. 
}

\begin{definition}[Components of a continuous-time hierarchical Markov
  model
  ]
\label{def:components}
A continuous-time hierarchical Markov model 
(with a two-level hierarchy) is specified by the components~$\hmm$:

\begin{compactitem}
\item An infinitesimal generator~$\modesM$ of a Markov model with
  initial distribution~$\pMnull$ with aggregates of states~$\Mi$,
  $i=1, \dots, \nM$. The~$\Mi$ are referred to as modes.
\item For each mode a Markov model with infinitesimal generator~$\Qi$
  and initial distribution~$\pQinull$
  .
\end{compactitem}
Then the infinitesimal generator~$\fullM$ of the aggregated model
for modal gating is calculated as follows:
\begin{equation}
  \label{eq:fullM}
  \fullM = \begin{pmatrix}
    \modesM^{1,1} \oplus Q^1 & | & \modesM^{1,2} \otimes P^{1,2} & | & \dots & | &
    \modesM^{1,\nM} \otimes P^{1,\nM} \\
    \hline
    \modesM^{2,1} \otimes P^{2,1} & | & \modesM^{2,2} \oplus Q^2 & |& \dots & | &
    \modesM^{2,\nM} \otimes P^{2,\nM} \\
    \vdots &  & & \ddots& & &\vdots\\
\vdots & 
& & & \ddots & & \vdots\\
    \vdots & & & & &\ddots & \vdots\\
    \modesM^{\nM,1} \otimes P^{\nM,1} & | & \dots  & \dots &\dots & | & \modesM^{\nM,\nM} \oplus Q^{\nM}
  \end{pmatrix}.
\end{equation}
\end{definition}

It is straightforward to generalise this definition recursively to an
arbitrary number of hierarchies. From Definition~\ref{def:parttens}
and~\eqref{eq:tensorpartition} we know that an arbitrary
distribution~$p$ over the states of the full model can be represented
by a linear combination of tensor products of the
form~\eqref{eq:tensorpartition}. We now require for initial
distributions that they should arise from a single tensor product of
initial distributions over the states of~$\modesM$ and initial
distributions over the states of the~$\Qi$.

\begin{definition}[Initial distribution over the states of a
  hierarchical Markov model]
  Let~$v_{\boldsymbol{m}}$ be the initial distribution over the states
  of the top-level model~$\modesM$ and~$w_{\boldsymbol{n}}$, a vector
  whose components~$w^i$ are initial distributions over the states of
  the models~$\Qi$. Then the initial distribution $\pnull$ over the
  states of the full model~$\fullM$ is calculated by the tensor
  product~`$\parttens$' introduced in Definition~\ref{def:parttens}:
  \begin{equation}
    \label{eq:state}
    \pnull=v_{\boldsymbol{m}} \parttens w_{\boldsymbol{n}} =(v^1 \otimes w^1; \dots; v^i \otimes w^i; \dots; v^{\nM} \otimes w^{\nM}).
  \end{equation}
  \label{def:initialdistribution}
\end{definition}

\begin{remark}
\label{rem:defstate}
We make some remarks regarding the interpretation of Definition~\ref{def:initialdistribution}:%
  \begin{compactitem}
  \item Note that whereas~$v_{\boldsymbol{m}}$ is a stochastic vector,
$w_{\boldsymbol{n}}$ is not.  It is easy to see that~$\pnull$ is a
stochastic vector.
\item Algebraically, Definition~\ref{def:initialdistribution}
  constrains initial distributions to so-called pure tensors which can
  be written as a single tensor product rather than a linear
  combination of tensor products.
\item Statistically, Definition~\ref{def:initialdistribution} says
  that for the initial distribution the probabilities of being in a
  state~$\modesMij$ and a state~$\Qik$ are stochastically independent:
  the joint probability of being in~$\modesMij$ and~$\Qik$ is the
  product of the individual probabilities~\eqref{eq:state}.

\end{compactitem}
\end{remark}

It is an interesting question if the time-dependent
solution~$\fullsol$ or the stationary distribution of the full
model~$\fullM$ remain in the
form~$\fullsol=\modessol\parttens w_{\boldsymbol{n}}(t)$ for $t>0$. In
fact, this is generally not the case.

\begin{remark}
  \textbf{Caution:} In most situations, $\fullsol$ cannot be written
  as a pure tensor
  $\fullsol=\modessol \parttens w_{\boldsymbol{n}}(t)$ for~$t>0$.  As
  discussed in Proposition~\ref{the:fullsolution} we obtain a
  solution~$(\modessol \parttens \statQsol)$ for a
  solution~$\modessol$ of~$\modesM$ and a vector~$\statQsol$ of
  stationary solutions~$\statQi$ of~$\Qi$ if and only if we choose
  initial conditions~$p^i=\statQi$ for all~$\Qi$.
\end{remark}

\subsection{Example}
\label{sec:modelex}

{As an example for the construction of the
  infinitesimal generator~$\fullM$ from the components $\hmm$ we
  present a model that will be used in Section~\ref{sec:results} for
  experimental data from the inositol trisphosphate receptor (\ipr).}

Let the infinitesimal generator for the switching between modes be

\begin{equation}
\label{eq:modesM}
\modesM=
\begin{pmatrix}
    -m_{13} & 0 & \vline & m_{13}\\
    0 & -m_{23} & \vline & m_{23}\\
    \hline
    m_{31} & m_{32} & \vline & -m_{31}-m_{32}
\end{pmatrix}
\end{equation}
and the models representing the intra-modal kinetics

{
\begin{equation}
  \label{eq:Qonetwo}
  Q^1 =
  \begin{pmatrix}
    -q^1_{12} & q^1_{12}\\
    q^1_{21} & -q^1_{21}
  \end{pmatrix} 
\text{ and }
  Q^2 =
  \begin{pmatrix}
    -q^2_{12} & q^2_{12}& 0 & 0\\
    q^2_{21} & -q^2_{21}-q^2_{23}-q^2_{24} & q^2_{23} & q^2_{24}\\
   0 & q^2_{32} & -q^2_{32} & 0\\
   0 & q^2_{42} & 0& -q^2_{42}
  \end{pmatrix} 
\end{equation}
with initial conditions
\begin{equation}
  \label{eq:pnull}
  p^1=(p^1_1, p^1_2) \text{ and } p^2=(p^2_1, p^2_2, p^2_3, p^2_4).
\end{equation}
}%
Then

\begin{align}
&  \fullM = \begin{pmatrix}
    \modesM^{1,1}_2 \oplus Q^1 & \vline & \modesM^{1,2}_2 \otimes P^{1,2}\\
    \hline
    \modesM^{2,1}_2 \otimes P^{2,1} & \vline & \modesM^{2,2}_2 \oplus Q^2
  \end{pmatrix} \nonumber \\
  \label{eq:exFullM}
         &  = 
           \left(\begin{smallmatrix}
             -m_{13} - q^1_{12} & q^1_{12} & 0 & 0 & \vline & m_{13}
             p^2_1 & m_{13} p^2_2 & m_{13} p^2_3 & m_{13} p^2_4\\
             q^1_{21} & -m_{13} - q^1_{21} & 0 & 0 & \vline & m_{13}
             p^2_1 & m_{13} p^2_2 & m_{13} p^2_3 & m_{13} p^2_4\\
             0 & 0 & - m_{23} - q^1_{12} & q^1_{12} &\vline & m_{23}
             p^2_1 & m_{23} p^2_2 & m_{23} p^2_3 & m_{23} p^2_4\\
             0 & 0 & q^1_{21} & - m_{23} - q^1_{21} & \vline & m_{23}
             p^2_1 & m_{23} p^2_2 & m_{23} p^2_3 & m_{23} p^2_4\\
             \hline
             m_{31} p^1_1 & m_{31} p^1_2 & m_{32} p^1_1& m_{32} p^1_2
             & \vline & -R-
             q^2_{12} & q^2_{12} & 0 & 0\\
             m_{31} p^1_1 & m_{31} p^1_2 & m_{32} p^1_1& m_{32} p^1_2
             & \vline & q^2_{21} & -R -
             q^2_{21}- q^2_{23} - q^2_{24}& q^2_{23} & q^2_{24}\\
             m_{31} p^1_1 & m_{31} p^1_2 & m_{32} p^1_1& m_{32} p^1_2& \vline & 0
             &q^2_{32}& -R - q^2_{32} & 0\\
             m_{31} p^1_1 & m_{31} p^1_2 & m_{32} p^1_1& m_{32} p^1_2 & \vline & 0 &
             q^2_{42}& 0 & -R- q^2_{42}
           \end{smallmatrix}
                           \right)
\end{align}
with~$R:=m_{31}+m_{32}$. 



\subsection{Parameterising the model with experimental data}
\label{sec:fit}

In order to parameterise the components $\hmm$ of our model, the
infinitesimal generators~$\modesM$ and~$\Qi$ have to be inferred from
ion channel data. We assume that the original data, a sequence of
current measurements recorded with a constant sampling interval~$\tau$
has been statistically analysed so that it has the form of
Figure~\ref{fig:changepoint}. Then each measurement has been
classified as open~(\gO) or closed~(\rC) and it has also been
determined in which mode~$\Mi$ the channel was at this point in
time. The Markov model~$\modesM$ is inferred from the sequence~$\seqM$
of modes~$\Mi$ whereas the models~$\Qi$ are parameterised from
sequences of~$\seqOC$ that are representative of a particular
mode. For example, in Figure~\ref{fig:changepoint}, the five data
points between~$j_{n}$ and~$j_{n+1}$ could be used for inferring the
model~\Qtwo\ representing the stochastic opening and closing within
mode~\Mtwo.

{All models are parameterised with the Bayesian method
  developed in \citet{Sie:11a, Sie:12b}. For inferring the
  infinitesimal generator~$\modesM$ the likelihood has the form

\begin{equation}
  \label{eq:likelihood}
  \prob((\seqM)|\modesM) = \statM \cdot P_{S^1} \cdot \exp (\modesM \tau) \cdot P_{S^2} \cdot ... \cdot \exp (\modesM \tau) \cdot P_{S^N} \cdot u^T,
\end{equation}
where~$(\seqM)$ is a sequence of observations of modes~$\Mi$ separated
by the sampling interval~$\tau$, $\modesM$ is the infinitesimal
generator of an aggregated Markov model, $\statM$ is the stationary
distribution of~$\modesM$ and~$u^T$ is a column vector of ones. The
matrices~$P_{\seqM}$ project to the states of the model that represent
the mode observed at data point~$k$. For example, 

\begin{equation}
\label{eq:project}
P_{\text{\Mone}} =
\begin{pmatrix}
  \id_{m_1} &\vline& 0& \vline& \dots &\vline&  0\\
  \hline
  0 &\vline& \dots\\
\vdots \\
0 &\vline& \dots &\vline& \dots &\vline& 0
\end{pmatrix}
\end{equation}
with the same block structure as in~\eqref{eq:Mmodel} projects to
states representing mode~\Mone, the other projection matrices
$P_{S^i}$ are defined equivalently. The likelihood for inferring the
infinitesimal generators~$\Qi$ from representative segments
of~$\seqOC$ of open~(\gO) and closed~(\rC) events
(Figure~\ref{fig:changepoint}) is analogous
to~\eqref{eq:likelihood}. See \citet{Sie:11a,Sie:12b} for a detailed
description of the
method.} 



\newpage

\section{Data-driven modelling of modal gating}
\label{sec:results}

{Our new framework enables us to easily construct and
  parameterise models for modal gating following a transparent
  iterative process:

  \begin{enumerate}
  \item Infer the stochastic process~$\seqM$ of switching between
    modes~$\Mi$ (Figure~\ref{fig:changepoint}) using the statistical
    method by \citet{Sie:14a}.
  \item Model the process~$\seqM$ of mode switching by parameterising
    an infinitesimal generator~$\modesM$
    (Figure~\ref{fig:switchingmodel}).
  \item From segments of~$\seqOC$ representative for the opening of
    closing within each of the modes~\Mone,
    \Mtwo, \dots\ (Figure~\ref{fig:kineticsmodel}) parameterise
    infinitesimal generators~\Qone, \Qtwo, \dots
  \item Choose initial distributions~$\pMnull$ and $p^i$ and
    combine all components $\hmm$ by calculating the infinitesimal
    generator~$\fullM$ of the full model (Figure~\ref{fig:modalamm}).
  \end{enumerate}

  Inferring~$\modesM$ and~$\Qi$ using the Bayesian approach briefly
  described in Section~
  \ref{sec:fit} ensures that the resulting model will be highly
  parsimonious because at each step a model with the optimal number of
  parameters for representing stochastic switching between modes, and
  opening and closing within modes, is determined.
} 
We demonstrate the practical implementation of this process using data
collected by \citet{Wag:12a} and compare the results with our
previously published model of the same data set \citep{Sie:12a}.

\subsection{Step (i): Statistical analysis of modal gating}
\label{sec:modalstats}

Previously, we have statistically analysed mode switching exhibited in
the data by \citet{Wag:12a} and found two modes, the nearly
inactive 
mode \Mone\ with a very low open probability and the highly
active 
mode \Mtwo\ with~$\Po\approx 70\%$, see \citet{Sie:14a} for
details. As illustrated in Figure~\ref{fig:changepoint} we have a
stochastic sequence of events~\Mone\ and \Mtwo\ that are separated by
a sampling interval~$\tau=\unit{0.05}{\milli \second}$.  We have
results from two types of the inositol trisphosphate
receptor~(\iprOne\ and \iprTwo) for various calcium concentrations
(\ca), \unit{0.01}{\micro M}, \unit{0.05}{\micro M} and
\unit{5}{\micro M}, at fixed concentrations of \unit{10}{\micro M}
inositol trisphosphate (\ipthree) and \unit{5}{\milli M} adenosine
trisphosphate (ATP). 
Empirical histograms of the sojourn times in~\Mone\ and \Mtwo\ for all
except one data set indicate that whereas time spent in the active
mode \Mtwo\ may be represented satisfactorily by one state, accurately
representing sojourn times in the nearly inactive mode~\Mone\ seems to
require at least two states, see Figure~\ref{fig:compParkDrive} for an
example. Whereas one state accounts for the support of the sojourn
time density in mode~\Mtwo\ (Figure~\ref{fig:IPR1Ca10nMDrive}) the
more widespread sojourn time density in mode~\Mone\ is better
approximated by two states (Figure~\ref{fig:IPR1Ca10nMPark}). Thus,
for five of our six data sets we parameterise~$\modesM$ with the
structure of~\eqref{eq:modesM}. For one data set (\iprTwo\ at
\unit{0.05}{\micro M} \ca), the histograms suggests that we need a
model with two states representing~\Mone\ and two states
representing~\Mtwo\ (Figure~\ref{fig:Ca50nMParkDrive}). Thus, for
these data we use the following infinitesimal generator:

\begin{equation}
  \label{eq:modesMCa50}
  \modesM_{\text{Type 2 \ipr, \unit{0.05}{\micro M} \ca}}=
\begin{pmatrix}
    -m_{12}- m_{12}& m_{12} & \vline & m_{13} & 0\\
    0 & -m_{24} & \vline & 0 & m_{24}\\
    \hline
    m_{31} & m_{32} & \vline & -m_{31} -m_{32} &0\\
    0 & m_{42} & \vline & 0 & -m_{42}
\end{pmatrix}.
\end{equation}

\subsection{Step (ii): Parameterising $\modesM$}
\label{sec:modesM}

Fitting~$\modesM$ to a time series~$\seqM$ of~\Mone\ and~\Mtwo\ using
our MCMC method \citep{Sie:11a, Sie:12b} is a challenging
problem. Because in a time series of a few hundred thousand up to
about a million data points the number of transitions between the two
modes is only in the order of hundreds, the data from which the rate
constants have to be inferred are effectively very limited---despite
the large number of data points. An example of a convergence plot
shown in Figure~\ref{fig:convergence} demonstrates that values of the
two rates, $m_{13}$ and~$m_{23}$, alternate. This is due to symmetry
in the model structure chosen for the model~$\modesM$ where the two
states~$M^1_1$ and~$M^1_2$ can be swapped without changing the
model. This effect can be removed by considering only one mode of the
multi-modal posterior, in this case by considering only samples
where~$m_{31}$ exceeds a certain threshold. Nevertheless, even after
this correction some parameters such as the rate~$m_{23}$ show a high
degree of uncertainty indicated by a widespread marginal distribution
(Figure~\ref{fig:convergence}). Mean values and standard deviations of
the distributions of the model parameters are summarised in
Tables~\ref{tab:ipr1} and~\ref{tab:ipr2}.

\begin{figure}[htbp]
  \raggedright
  \begin{minipage}{0.45\linewidth}
  \subfloat[]{%
    \includegraphics[width=
    \linewidth]{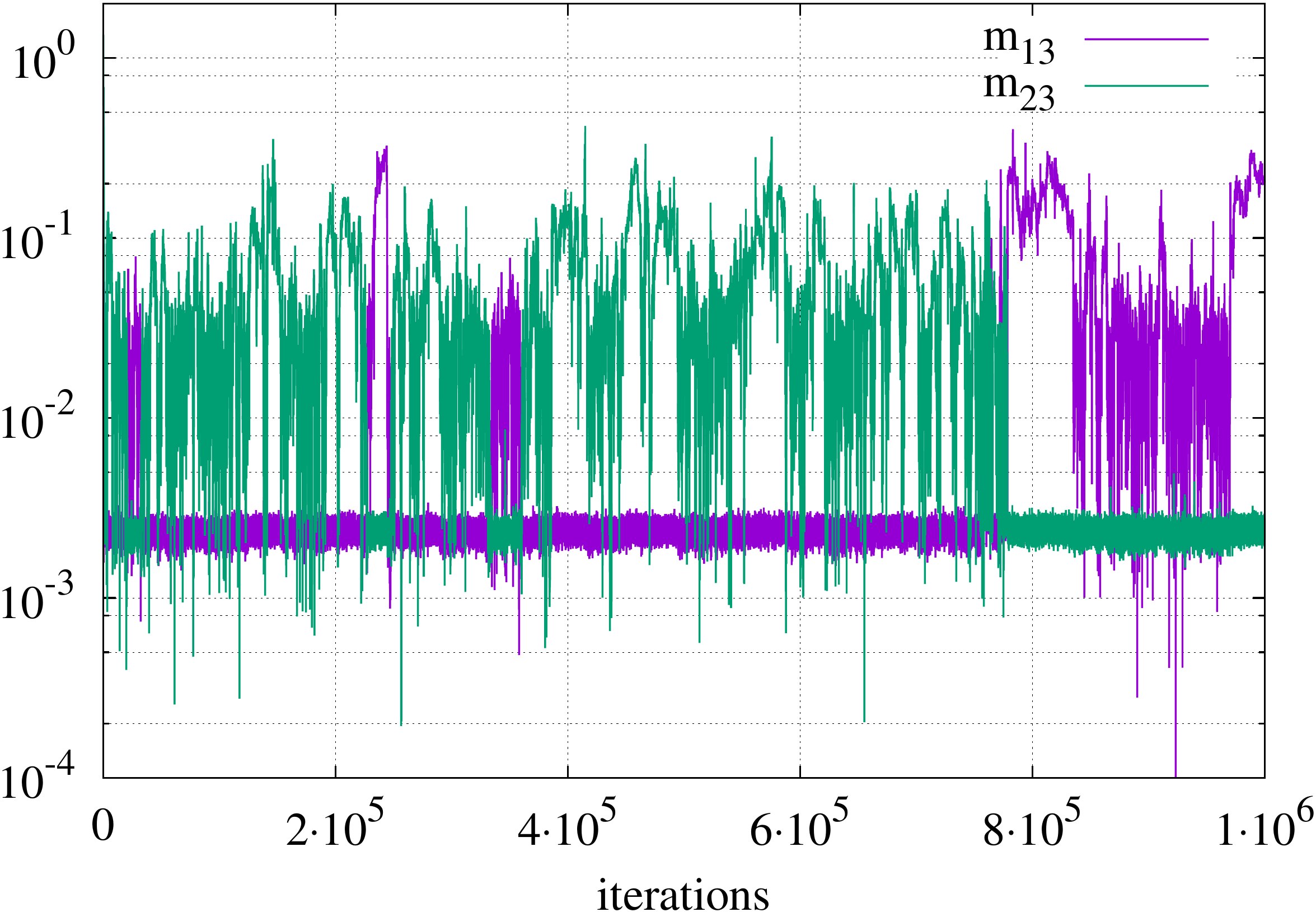}%
  }%
  \end{minipage}
\;
\begin{minipage}{0.45\linewidth}

  \subfloat[]{%
    \includegraphics[width=
    \linewidth]{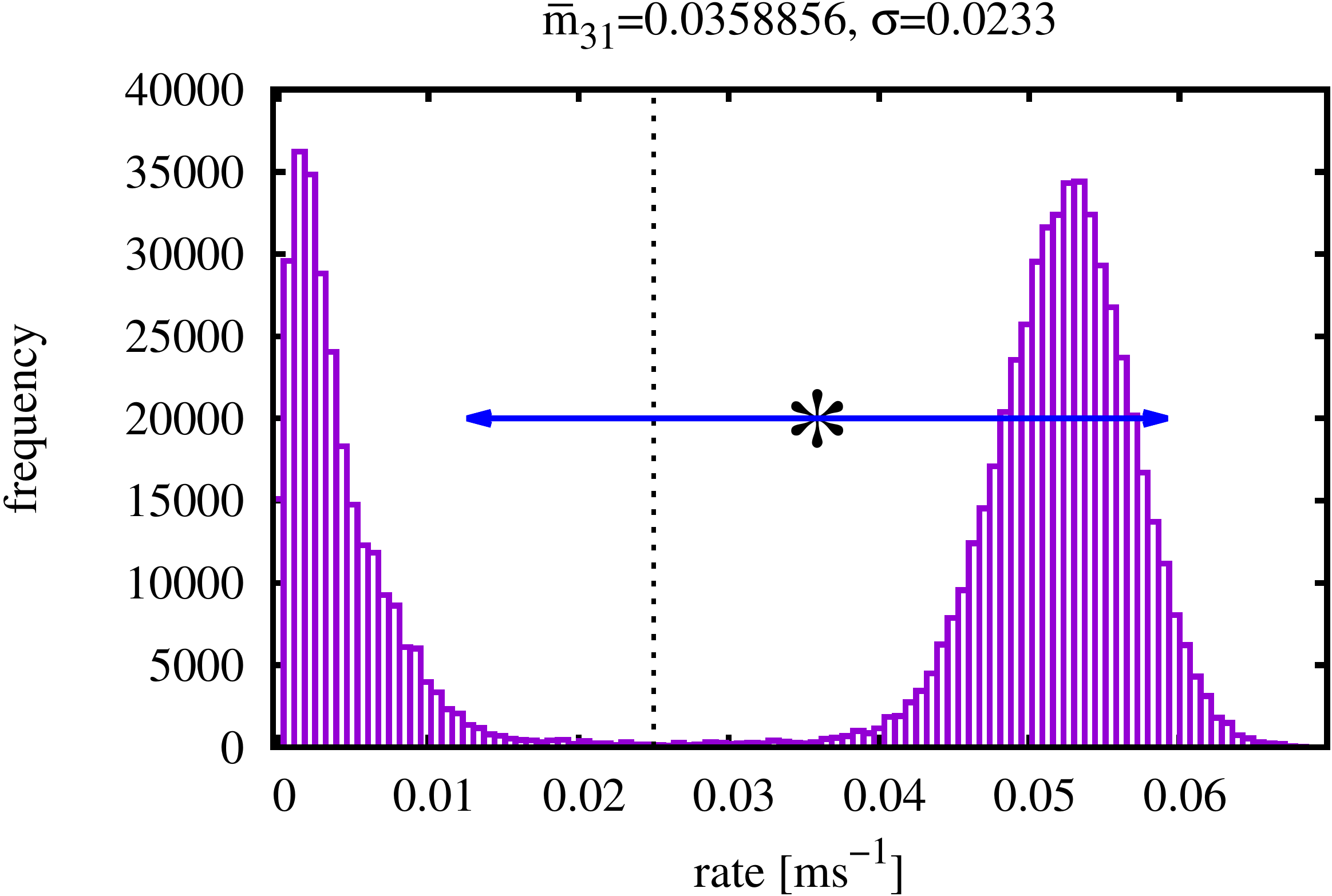}%
}%
\end{minipage}
\\
\medskip

\begin{minipage}{0.45\linewidth}
  \subfloat[]{%
    \includegraphics[width=\linewidth]{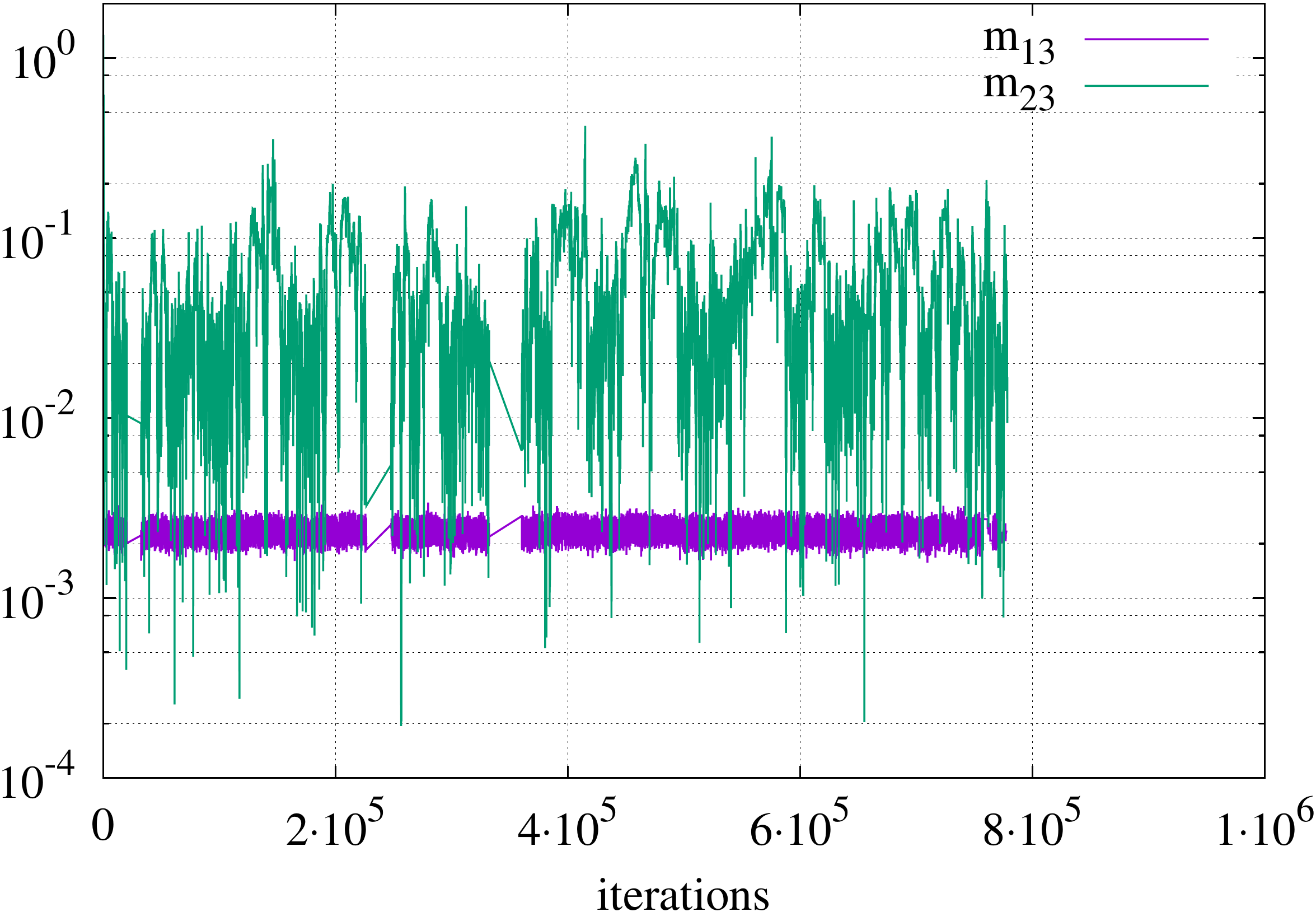}%
  }%
\end{minipage}
\;
\begin{minipage}{0.45\linewidth}
  \subfloat[]{%
    \includegraphics[width=\linewidth]{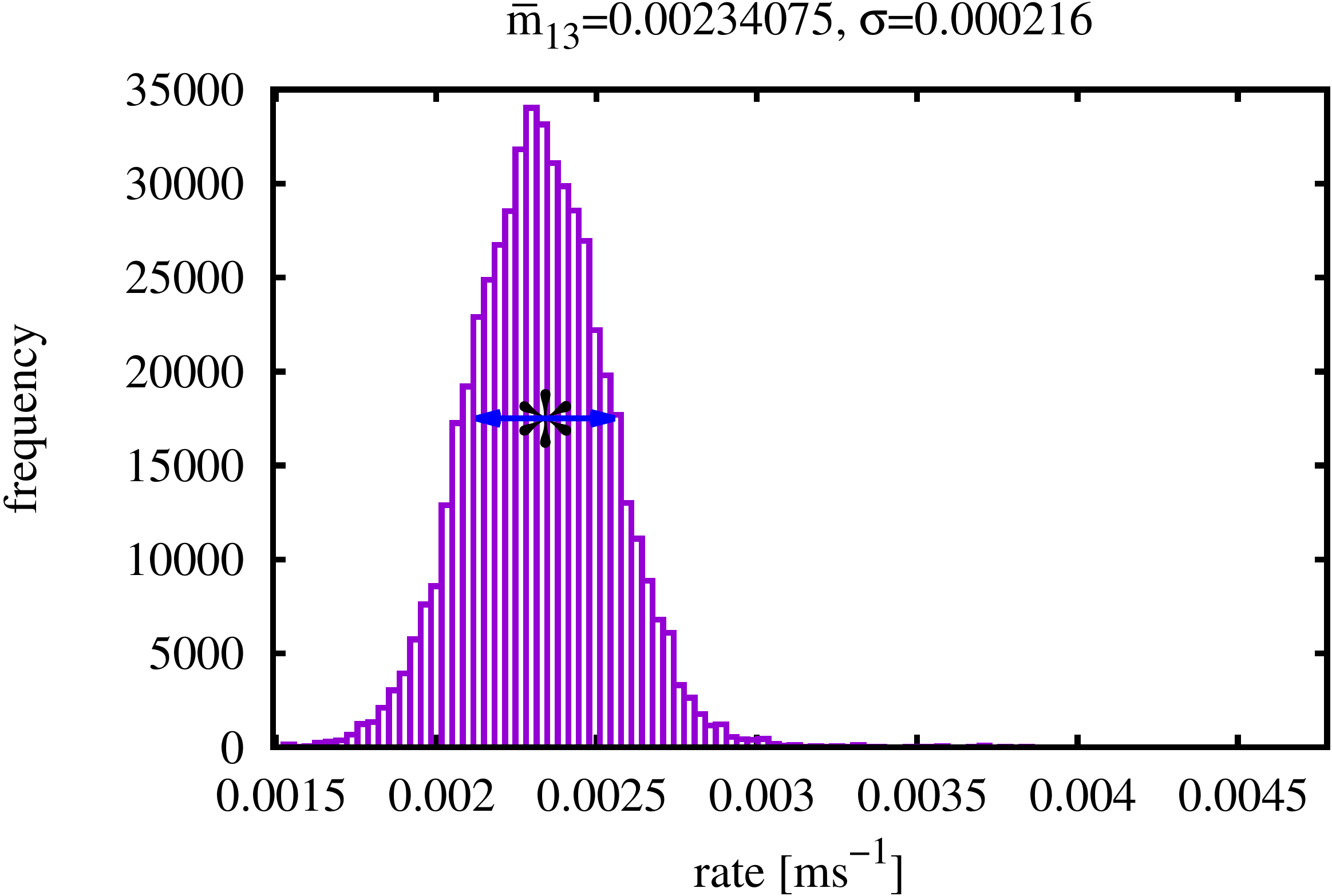}%
  }%
 \\
  \subfloat[]{%
    \includegraphics[width=\linewidth]{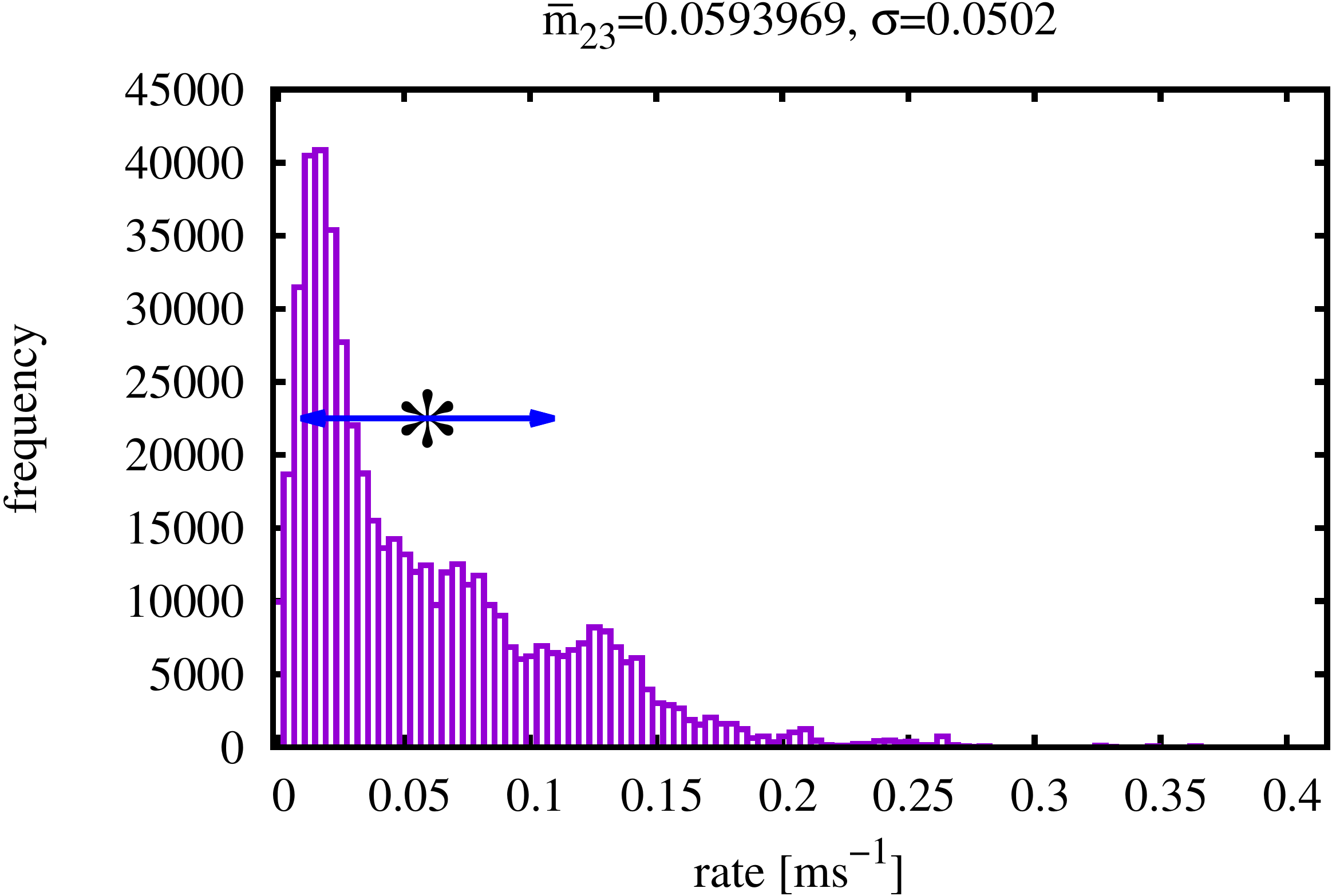}%
  }%
\end{minipage}
\caption{Fitting the stochastic transitions between modes~\Mone\
  and~\Mtwo\ to a 
  model~$\modesM$~\eqref{eq:modesM} with two states representing the
  nearly inactive mode~\Mone\ and one state representing the active
  mode~\Mtwo\ is challenging. The MCMC sampler by \citet{Chr:10a} is
  run with default parameters in order to generate samples from the
  posterior density of the Bayesian model described in
  Section~
  \ref{sec:fit} \cite{Sie:11a,Sie:12b}. The convergence plot in (a)
  shows that over the course of $10^6$ iterations the samples
  generated for the rates that enter the state~$M^2_3$ are
  occasionally swapped. This phenomenon is known as label switching in
  the MCMC literature and is caused by the symmetry of the
  model~$\modesM$~\eqref{eq:modesM}. The marginal histogram of the
  rate~$m_{31}$ is bimodal with two well separated peaks (b) so that
  the effect of label switching can be removed by discarding samples
  with 
  $m_{31}<\unit{0.025}{\milli \reciprocal{\second}}$ (indicated by a
  vertical line). The convergence plot obtained after thresholding is
  shown in (d). Whereas the marginal histogram~(e) indicates
  that~$m_{13}$ is well-constrained, the standard deviation
  of~$m_{23}$ remains high (f). Burn-in for all histograms:
  $2 \cdot 10^5$ iterations.}
  \label{fig:convergence}
\end{figure}

\begin{table}[htbp]
  \centering
  \begin{tabular}{lll}
    Type~I~\ipr\\
    \toprule
 \ca [\micro M]   &    \begin{minipage}[c]{0.25\linewidth}
      \centering
      $m_{13}$\\
      $m_{23}$\\
    \end{minipage} & \begin{minipage}[c]{0.25\linewidth}
      \centering
      $m_{31}$\\
      $m_{32}$\\
    \end{minipage}\\
    \cmidrule(lr){2-3}
    & 0.00236708 $\pm$ 0.000201138 & 0.0545511 $\pm$ 0.00294464 \\
0.01 & 0.069589 $\pm$ 0.0510011 &
    0.00318407 $\pm$ 0.00203495 \\ 
    \midrule
      & 0.0020581 $\pm$	0.000389371 & 0.0619096 $\pm$	0.00894936\\
0.05 & 0.01873 $\pm$	0.0111274 &  0.0101597 $\pm$	0.00693104\\
    \midrule
      & 0.00311881 $\pm$ 0.00027425 &
                                0.0564093 $\pm$ 0.00404197\\
5      & 0.160984 $\pm$ 0.06707 & 0.00472598 $\pm$ 0.00189248 \\
    \bottomrule
  \end{tabular}
  \caption{Type I~\ipr: Mean values and standard deviations for the rate constants of the infinitesimal generator~$\modesM$~\eqref{eq:modesM} in the main text representing the transitions between an inactive mode~\Mone\ and an active mode~\Mtwo. All values are given in transitions per milliseconds [\milli\reciprocal\second].}
  \label{tab:ipr1}
\end{table}

\begin{table}[htbp]
  \centering
  \begin{tabular}{lll}
    Type~II~\ipr\\
    \toprule
 \ca [\micro M]   &    \begin{minipage}[c]{0.25\linewidth}
      \centering
      $m_{13}$\\
      $m_{23}$\\
    \end{minipage} & \begin{minipage}[c]{0.25\linewidth}
      \centering
      $m_{31}$\\
      $m_{32}$\\
    \end{minipage}\\
    \cmidrule(lr){2-3}
    & 0.00134665 $\pm$	0.000250273 & 0.0724154 $\pm$	0.00817478 \\ 
0.01      & 0.0714618 $\pm$ 0.0454381 & 0.0139203 $\pm$	0.00588837\\
    \midrule
   & 0.00435935 $\pm$	0.00027004 & 0.0284326 $\pm$	0.00151654 \\
    5& 0.146953 $\pm$	0.0424637 & 0.00230764 $\pm$	0.000712072\\
    \midrule
      &    \begin{minipage}[c]{0.25\linewidth}
      \centering
      $m_{12}$\\
      $m_{13}$\\
      $m_{24}$\\
    \end{minipage} & \begin{minipage}[c]{0.25\linewidth}
      \centering
      $m_{21}$\\
      $m_{31}$\\
      $m_{42}$\\
    \end{minipage}\\
    \cmidrule(lr){2-3}
      &     0.00112896 $\pm$	0.000402963 & 0.0732717 $\pm$	0.034058 \\
    0.05 & 0.000756359 $\pm$	0.000133923 & 0.083959 $\pm$ 0.0184139 \\
      & 0.0451628 $\pm$	0.0168949 &  0.001816 $\pm$	0.000335203\\
    \bottomrule
  \end{tabular}
  \caption{Type II~\ipr: Mean values and standard deviations for the
    rate constants of the infinitesimal generator~$\modesM$, \eqref{eq:modesM}  and~\eqref{eq:modesMCa50}, respectively, representing the transitions between an inactive mode~\Mone\ and an active mode~\Mtwo. For \unit{0.05}{\micro M} \ca\ an additional state was required for representing the dynamics of the active mode~\Mtwo. All values are given in transitions per milliseconds [\milli\reciprocal\second].}
  \label{tab:ipr2}
\end{table}

\subsection{Step (iii): Parameterising \Qone\ and \Qtwo}
\label{sec:Qi}

In our previous study \citet{Sie:12a} we have already fitted a model
with two states to representative segments of the inactive mode~\Mone\
and a model with four states for representing~\Mtwo,
see~\eqref{eq:Qonetwo} for the form of the infinitesimal
generators~\Qone\ and~\Qtwo. Interestingly, we could show that~\Qone\
and~\Qtwo\ were independent of the concentrations of \ipthree, ATP and
\ca. The parameter values from the Supplementary Material of
\citet{Sie:12a} are reproduced here for convenience
(Table~\ref{tab:intramodal}).

\begin{table}[htbp]
  \centering
  \begin{tabular}{lll}
    \toprule
    & \multicolumn{2}{c}{\Mone}\\
     \cmidrule(lr){2-3}
     &    \begin{minipage}[c]{0.25\linewidth}
      \centering
      $q^1_{12}$\\
      \bigskip
    \end{minipage} & \begin{minipage}[c]{0.25\linewidth}
      \centering
      $q^1_{21}$\\
      \bigskip
    \end{minipage}\\
    \midrule
    Type~I \ipr\ & 11.1 $\cdot 10^{-3} \pm$ 1.01 $\cdot  10^{-3}$ &
    3.33 $\pm$ 0.27\\
    \hline
    Type~II \ipr\ & 4.14 $\cdot 10^{-3} \pm$ 6.7 $\cdot 10^{-4}$ & 3.42
    $\pm$ 0.496 \\
    \midrule

    & \multicolumn{2}{c}{\Mtwo}\\
     \cmidrule(lr){2-3}
    &    \begin{minipage}[c]{0.25\linewidth}
      \centering
      $q^2_{12}$\\
      $q^2_{23}$\\
      $q^2_{24}$\\
      \bigskip
    \end{minipage} & \begin{minipage}[c]{0.25\linewidth}
      \centering
      $q^2_{21}$\\
      $q^2_{32}$\\
      $q^2_{42}$\\
      \bigskip
    \end{minipage}\\
    \midrule
    & 1.24 $\pm$ 0.121 & 0.0879 $\pm$ 0.0117\\
    Type I \ipr\ & 3.32 $\cdot 10^{-3} \pm$ 1.64 $\cdot 10^{-3}$ &
    0.0694 $\pm$ 0.0266\\ 
    & 10.5 $\pm$ 0.0771 & 4.01 $\pm$ 0.0293 \\
    \midrule
    & 1.14 $\pm$ 0.0956 & 0.0958 $\pm$ 0.00945\\
    Type II \ipr\ & 4.75 $\cdot 10^{-3} \pm$ 1.53 $\cdot 10^{-3}$ &
    0.0119 $\pm$ 0.00357\\ 
    & 10.1 $\pm$ 0.0668 & 3.27 $\pm$ 0.0221 \\
    \bottomrule
  \end{tabular}
  \caption{Mean values and standard deviations for rate constants of
    the infinitesimal generators~\eqref{eq:Qonetwo} representing
    opening and closing in~\Mone\ and~\Mtwo. All values are given in
    transitions per milliseconds [\milli\reciprocal\second]. Reproduced from the Supplementary Material of~\citet{Sie:12a}.}
  \label{tab:intramodal}
\end{table}

\subsection{Step (iv): The generator~$\fullM$ of the full model}
\label{sec:fullM}

After the models~$\modesM$, \Qone\ and \Qtwo\ have been obtained, we
finally need to specify the initial distributions~$\pMnull$, $p^1$
and~$p^2$. Consistent with the experimental assumption that recording
of the data was started when the channel has reached steady state we
set~$\pMnull=\statM$, $p^1=\statQ^1$ and~$p^2=\statQ^2$
where~$\statM$, $\statQ^1$ and~$\statQ^2$ are the stationary
distributions of~$\modesM$, \Qone\ and~\Qtwo,
respectively. 
After all components $\hmm$ of our model have been specified, the
infinitesimal generator~$\fullM$ of the full model can be calculated
using~\eqref{eq:fullM}.

\subsection{Results}
\label{sec:comparison}

{Due to the problems with fitting the infinitesimal
  generator~$\modesM$~\eqref{eq:modesM} mentioned in
  Section~\ref{sec:modesM} one may ask if a simpler two-state model
  representing the dynamics of modal gating would be preferable.}
However, the ability of a three-state model to approximate the sojourn
distribution of the nearly inactive mode~\Mone\ more
accurately~(Figure~\ref{fig:IPR1Ca10nMPark}) was found to be crucial
for obtaining a better fit of the closed time distribution in
comparison with the model from
\cite{Sie:12a}~(Figure~\ref{fig:IPR1Ca10nMClosed}). That the model
structure of the hierarchical model proposed here is better able to
capture the properties of the entire time series data seems even more
convincing because it has---unlike the original model from
\citet{Sie:12a}--- been built without directly fitting to the time
series at any step of its construction.

\begin{figure}[htbp]
  \centering
  \subfloat[sojourn time density in \Mone]{%
    \label{fig:IPR1Ca10nMPark}
    \includegraphics[width=0.45\textwidth]{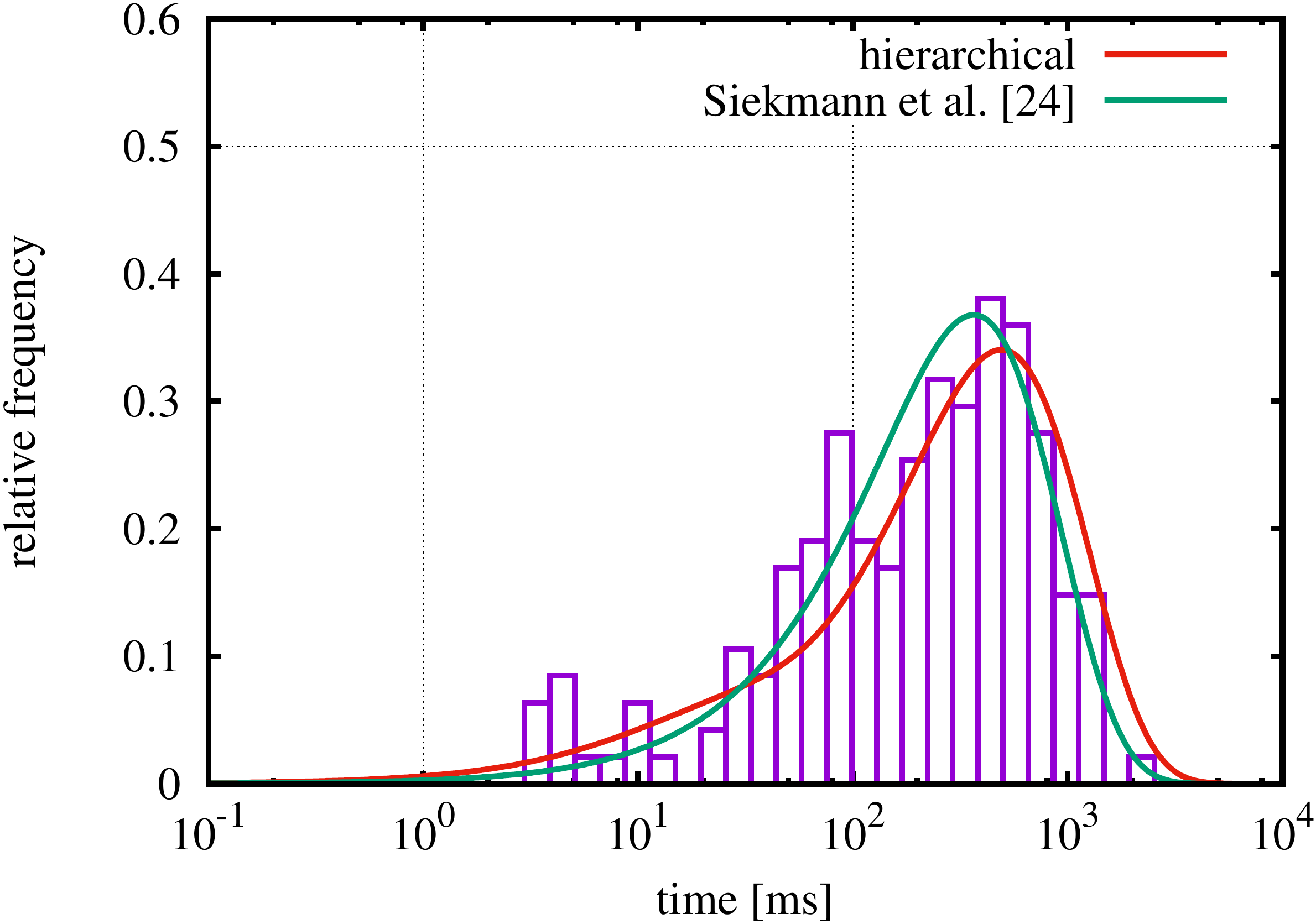}%
  }%
\;
  \subfloat[sojourn time density in \Mtwo]{%
    \label{fig:IPR1Ca10nMDrive}
    \includegraphics[width=0.45\textwidth]{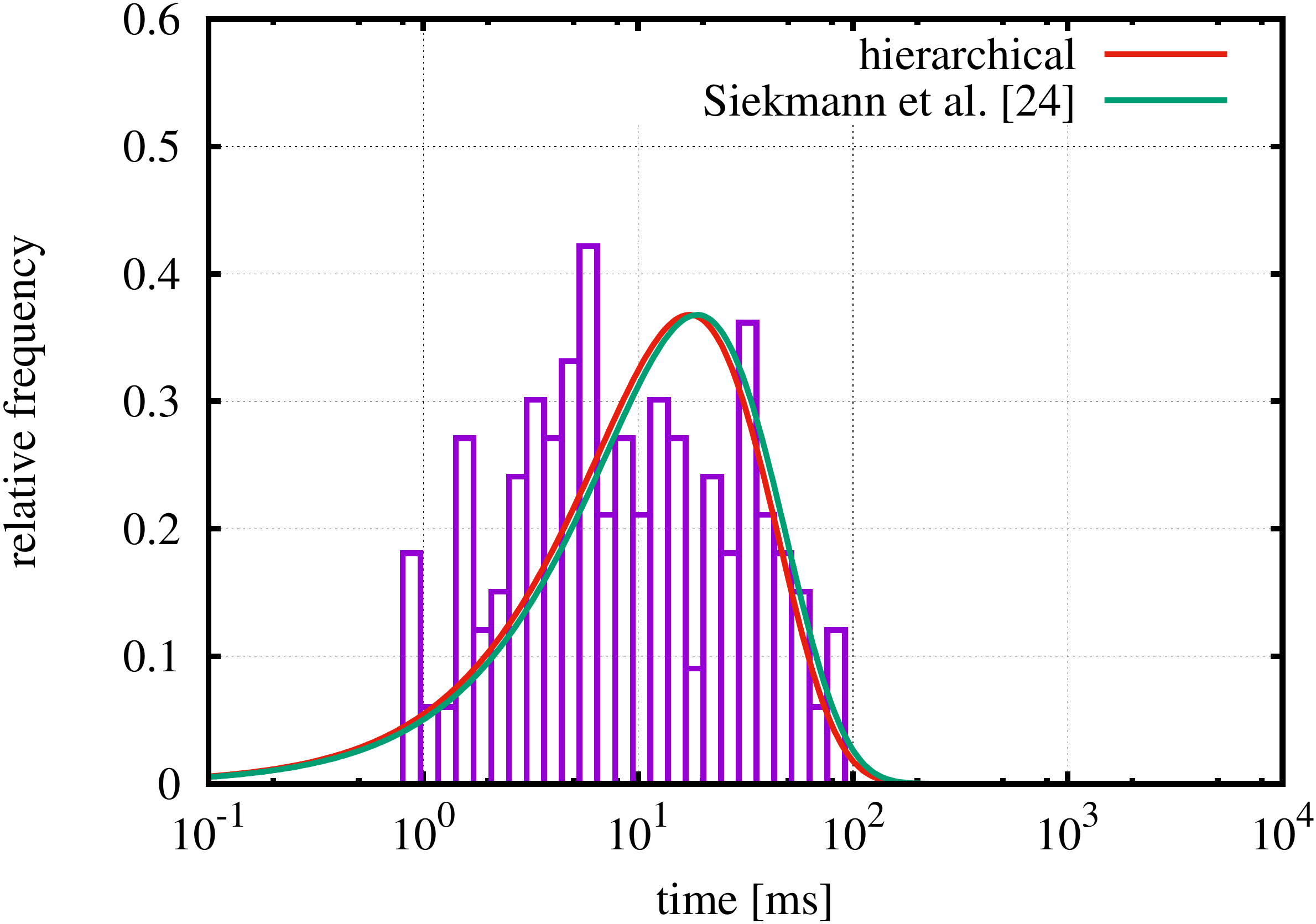}%
  }%
\\
  \subfloat[closed time density]{%
    \label{fig:IPR1Ca10nMClosed}
    \includegraphics[width=0.45\textwidth]{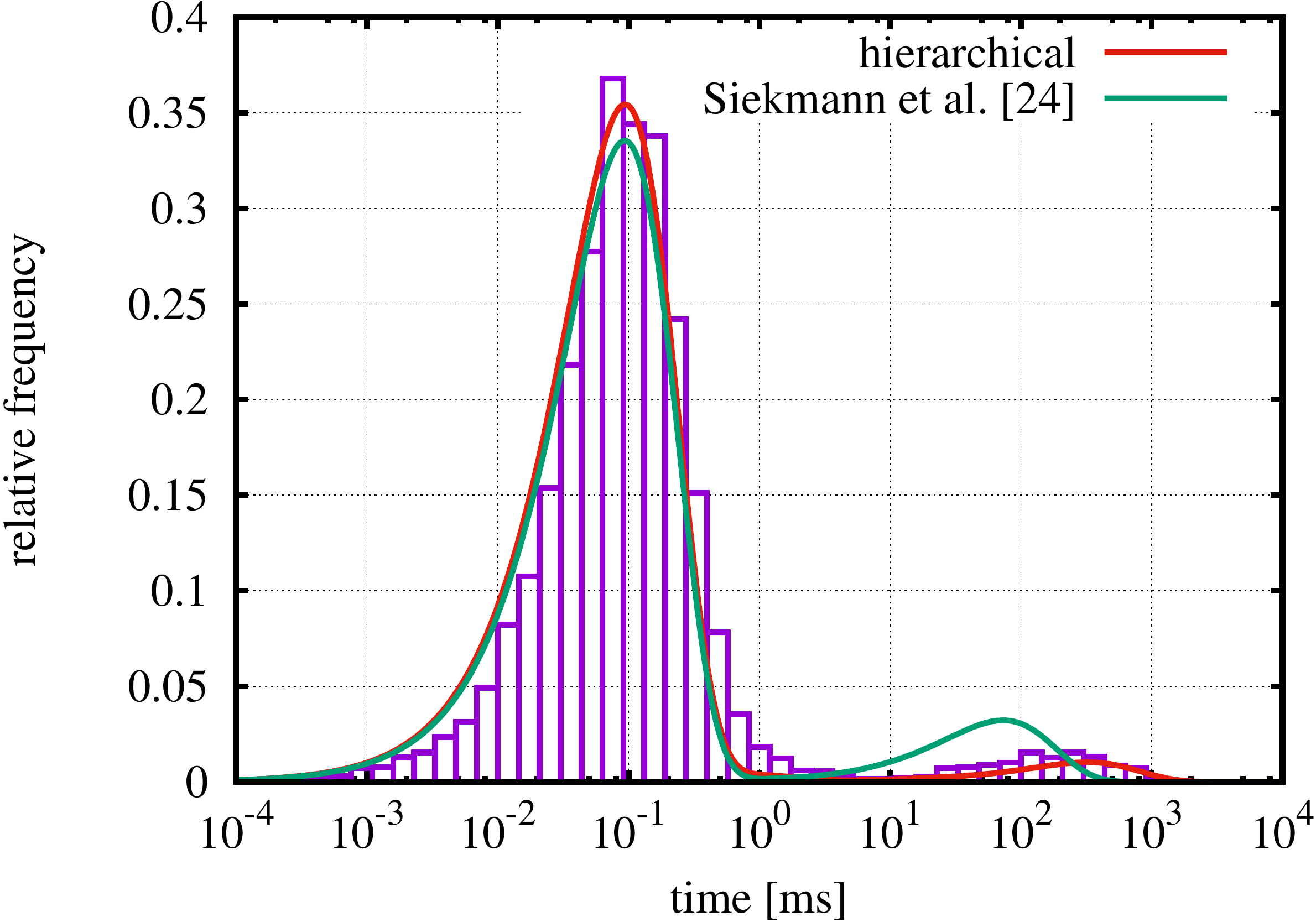}%
  }%
\;
  \subfloat[open time density]{%
    \label{fig:IPR1Ca10nMOpen}
    \includegraphics[width=0.45\textwidth]{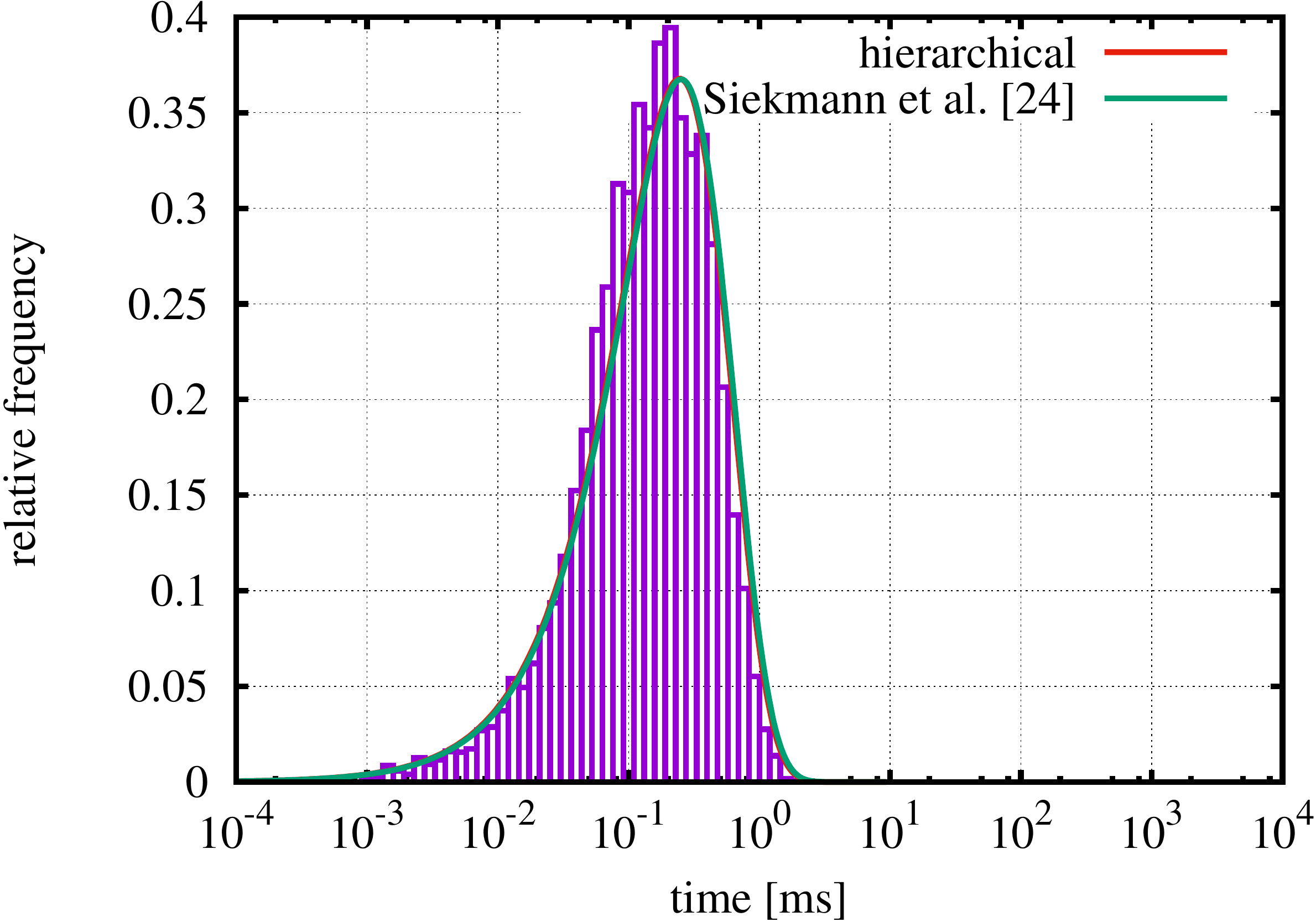}%
  }%
  \caption{The model from \citet{Sie:12a} and the new hierarchical
    model are compared for a data set from~\iprOne~for
    \unit{10}{\micro M} \ipthree, \unit{5}{\milli M} ATP
    and~\unit{0.01}{\micro M} \ca. (a) shows that the fit of the new
    model to the empirical sojourn time density in mode~\Mone~(shown
    in red) is slightly improved in comparison with the original
    model~(shown in green). This improved fit of the modal kinetics
    clearly improves the fit to the closed time densities shown
    in~(c).}
  \label{fig:compParkDrive}
\end{figure}

In Figure~\ref{fig:modalclosed} we show that the bimodal closed time
distribution observed for some combinations of ligand concentrations
arises due to the mixing of the closed time distributions within
nearly inactive mode~\Mone\ and active mode~\Mtwo\ both of which only
have one distinct maximum.

\begin{figure}[htbp]
  \centering
  \includegraphics[width=0.6\textwidth]{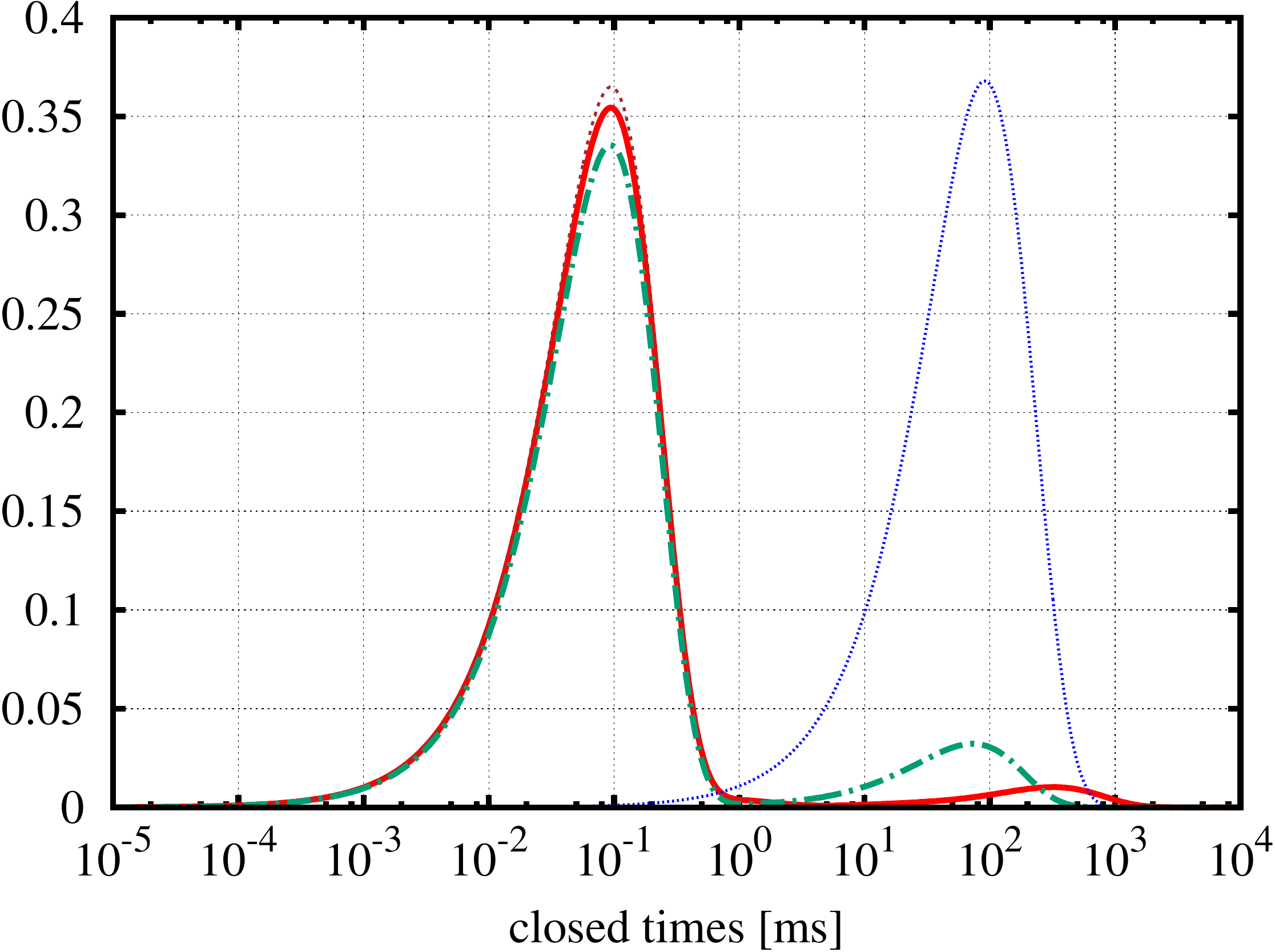}%
  \caption{The bimodal closed time distribution (red, solid) observed
    for~\iprOne~for \unit{10}{\micro M} \ipthree, \unit{5}{\milli M}
    ATP and~\unit{0.01}{\micro M} \ca\
    (Figure~\ref{fig:IPR1Ca10nMClosed}) arises due to mixing of the
    closed time distribution of the active mode \Mtwo\ (brown, dashed)
    and of the inactive mode~\Mone\ (blue, dotted). Note that the mode
    of the closed time distribution in \Mone\ (blue, dotted) is
    shifted from about \unit{100}{\milli \second} to \unit{300}{\milli
      \second} in the closed time distribution of the full model. By
    comparison with the closed time distribution of the model from
    \citet{Sie:12a} (green, dash-dotted) it shows that this model is
    incapable of shifting the mode of the closed distribution in the
    nearly inactive mode \Mone\ to the right.}
  \label{fig:modalclosed}
\end{figure}

Stronger differences between both models are observed for a data set
collected from~\iprTwo~for \unit{10}{\micro M} \ipthree,
\unit{5}{\milli M} ATP and~\unit{0.05}{\micro M} \ca. For this
experimental condition, the effect of modal gating can be observed
without statistical
analysis~(Figure~\ref{fig:IPR2Ca50nMdata}). Figure~\ref{fig:Ca50nMParkDrive}
shows that both modes~\Mone\ and~\Mtwo\ exhibit a widespread
distribution of sojourn times which can only approximately be captured
by a four-state model with two states each for both~\Mone\
and~\Mtwo. Whereas the new hierarchical model can approximate the
empirical distributions of both modes relatively well, the model from
\citet{Sie:12a} fails due to the fact that only one characteristic
sojourn time for each mode can be captured by the pair of transition
rates accounting for modal gating in this
model~(Figure~\ref{fig:Ca50nMParkDrive}).

\begin{figure}[htbp]
  \centering
  \subfloat[sojourn time distribution in \Mone]{%
    \includegraphics[width=0.45\textwidth]{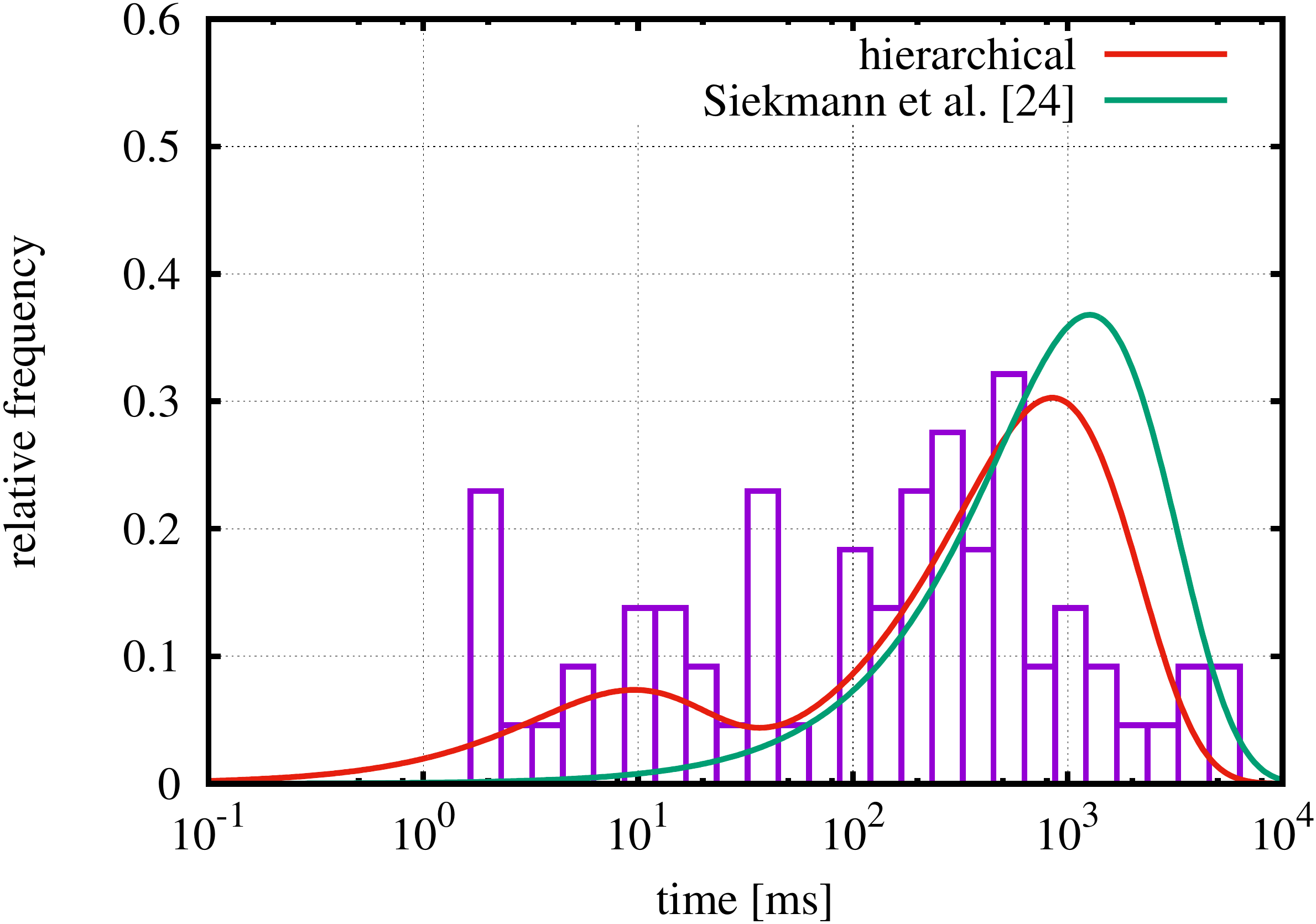}%
  }%
\;
  \subfloat[sojourn time distribution in \Mtwo]{%
    \includegraphics[width=0.45\textwidth]{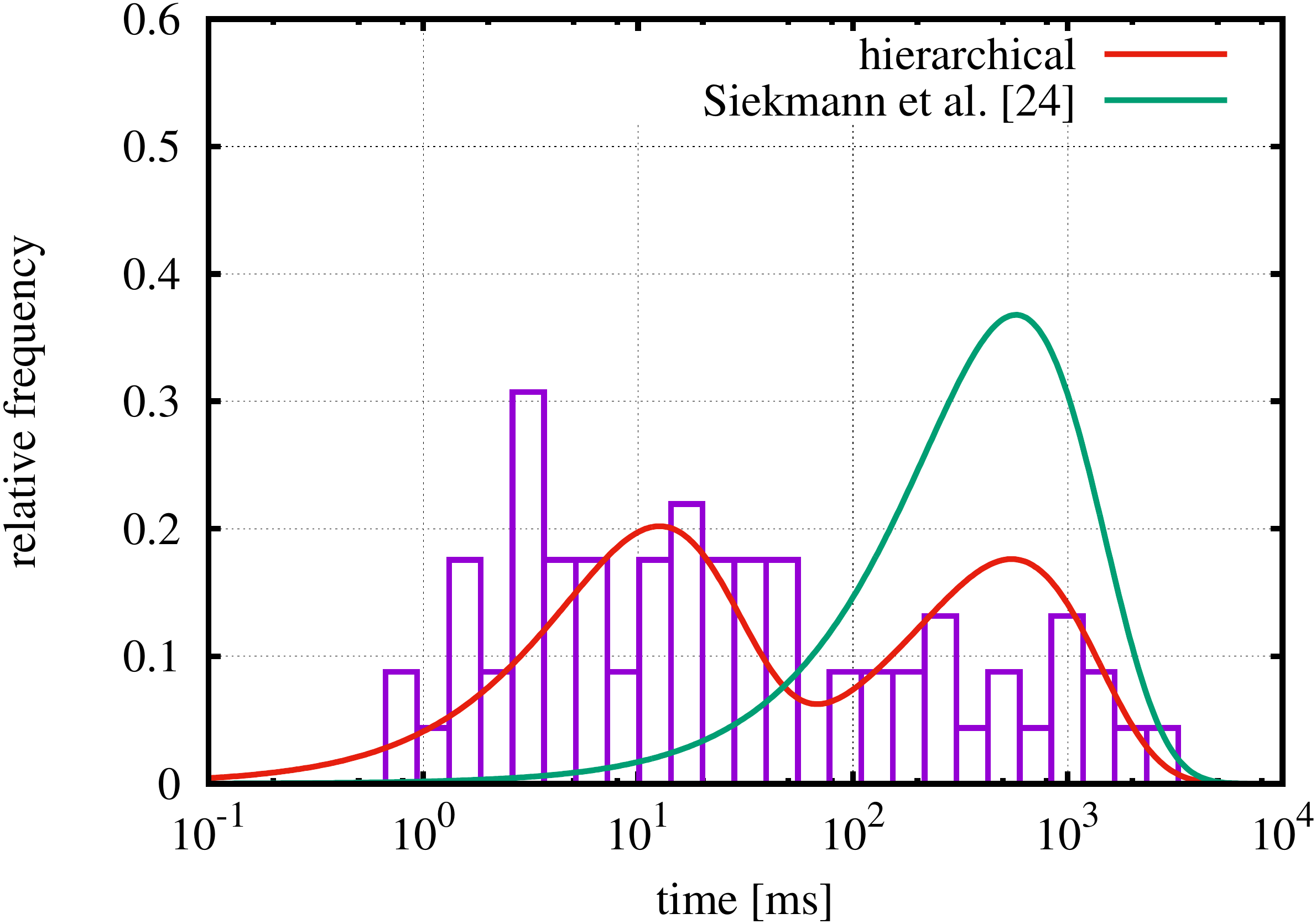}%
  }%
  \caption{Empirical sojourn time distributions for both modes~\Mone\
    and~\Mtwo\ for~\iprTwo~for~for \unit{10}{\micro M} \ipthree,
    \unit{5}{\milli M} ATP and~\unit{0.05}{\micro M} \ca. Whereas the
    hierarchical model can resolve (by using a four-state model) the
    widespread distributions of both~\Mone\ and~\Mtwo, the model from
    \citet{Sie:12a} can only capture one characteristic sojourn time
    due to the fact that only one pair of transition rates has been
    used to connect the submodels for mode~\Mone\ and~\Mtwo.}
  \label{fig:Ca50nMParkDrive}
\end{figure}

Due to the failure to account for the modal sojourn time
distributions, we expect the model from \citet{Sie:12a} to reproduce
the kinetics observed in the data much less accurately than the new
hierarchical model. In order to illustrate this we simulated both the
\citet{Sie:12a} model~(Figure~\ref{fig:IPR2Ca50nMold}) and the new
model~(Figure~\ref{fig:IPR2Ca50nMnew}). The sample path was plotted in
blue when the channel was in mode~\Mone\ whereas it was plotted in
brown when the channel was in mode~\Mtwo. The same colours were used
for colouring the data~(Figure~\ref{fig:IPR2Ca50nMdata}) based on the
results of the statistical analysis from \citet{Sie:14a}. The
comparison shows that mode switching happens much more frequently in
the model from \citet{Sie:12a} than observed in the data and the
proportion of relatively long sojourns is increased with respect to
the data. The frequency of mode switching and the widespread
distribution of sojourn lengths is better approximated by the
hierarchical model. The burst of activity observed in the data
starting after approximately~\unit{45}{\second} is not captured by
either model. This would require separate statistical analysis of the
trace before and after~$t=\unit{45}{\second}$, the observed change of
activity.

\begin{figure}[htbp]
  \centering
  \subfloat[data]{%
    \label{fig:IPR2Ca50nMdata}
    \includegraphics[width=0.45\textwidth]{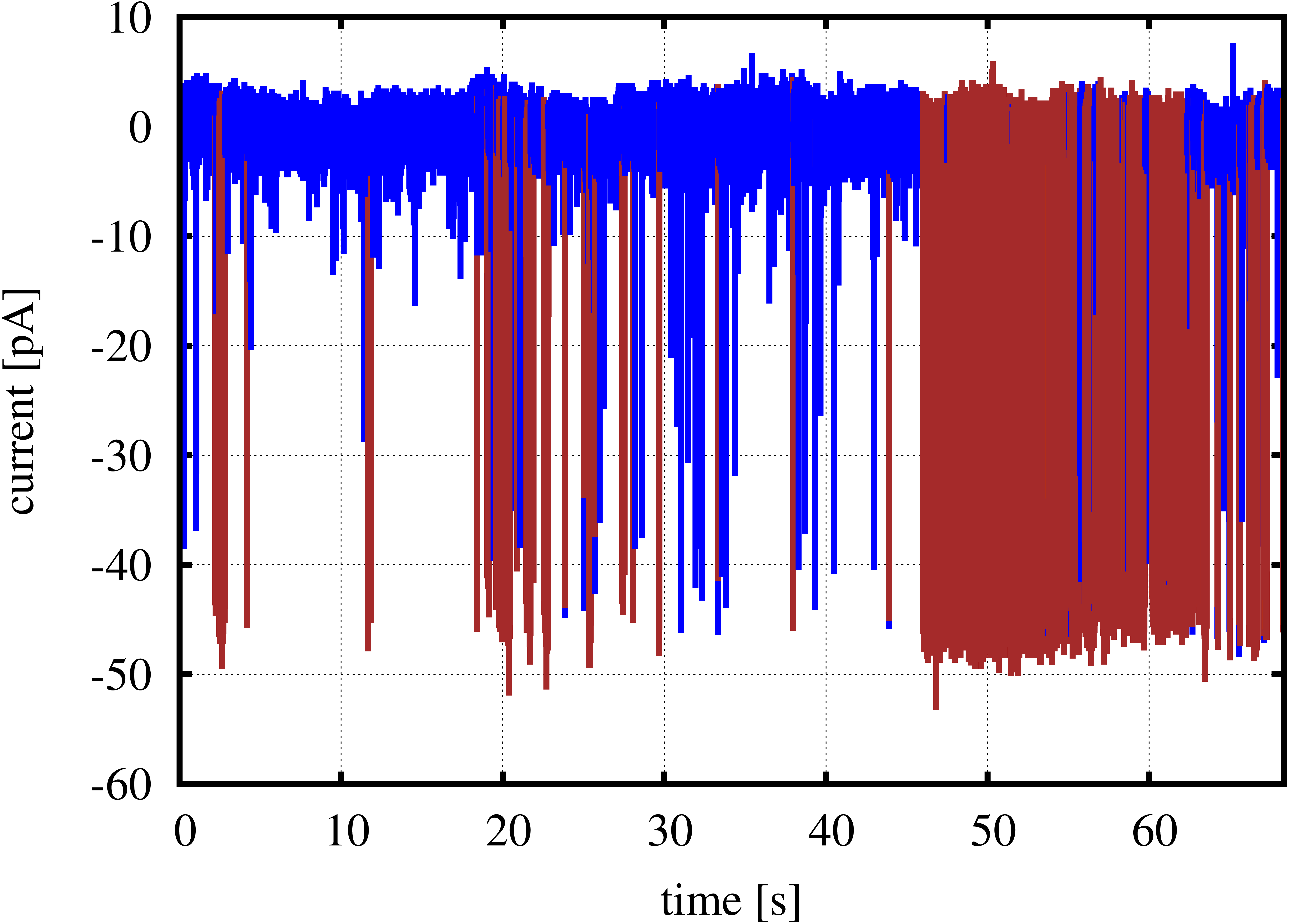}%
  }%
\\
  \subfloat[Hierarchical model]{%
    \label{fig:IPR2Ca50nMnew}
    \includegraphics[width=0.45\textwidth]{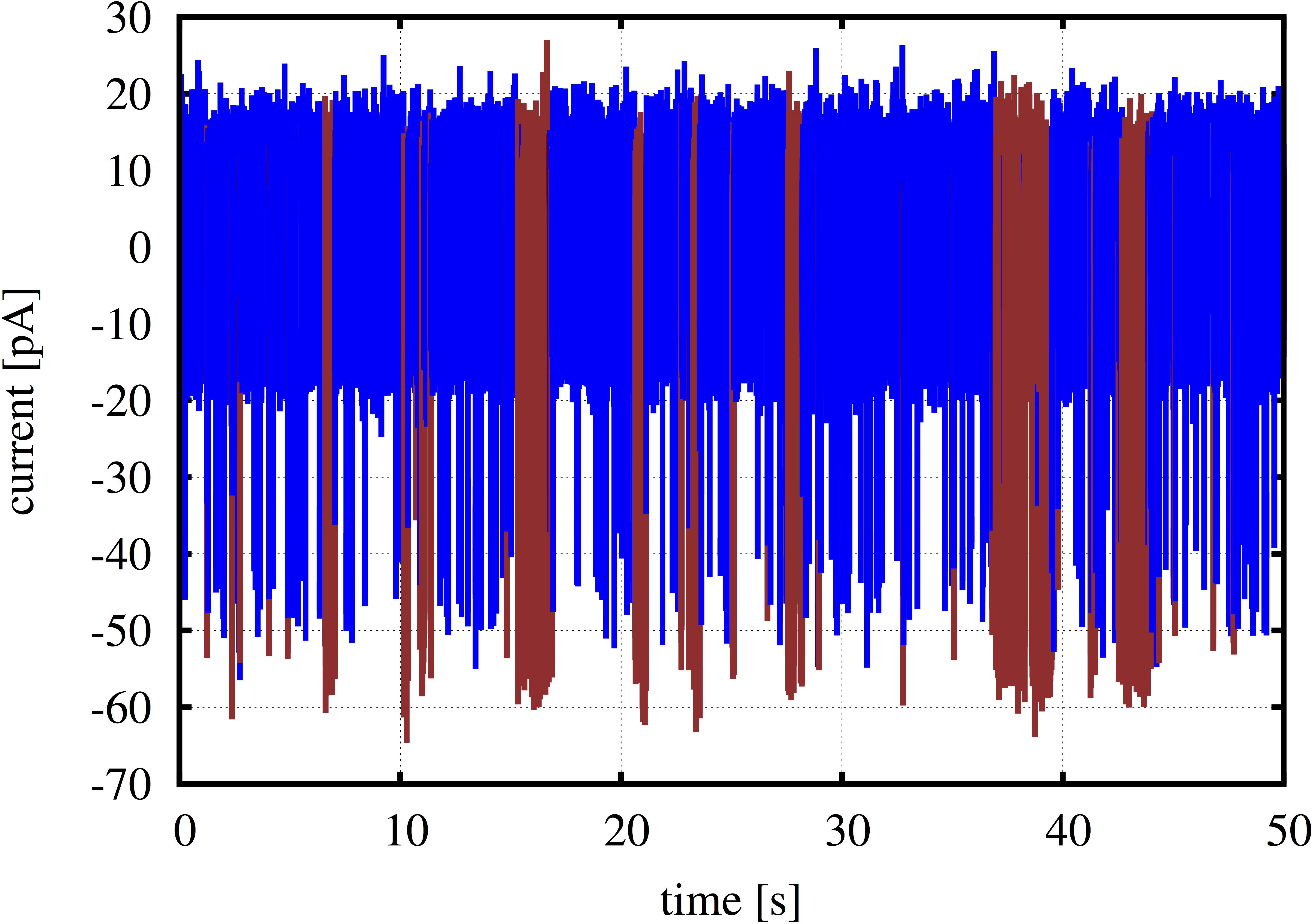}%
  }%
  \;
  \subfloat[\citet{Sie:12a} model]{%
    \label{fig:IPR2Ca50nMold}
    \includegraphics[width=0.45\textwidth]{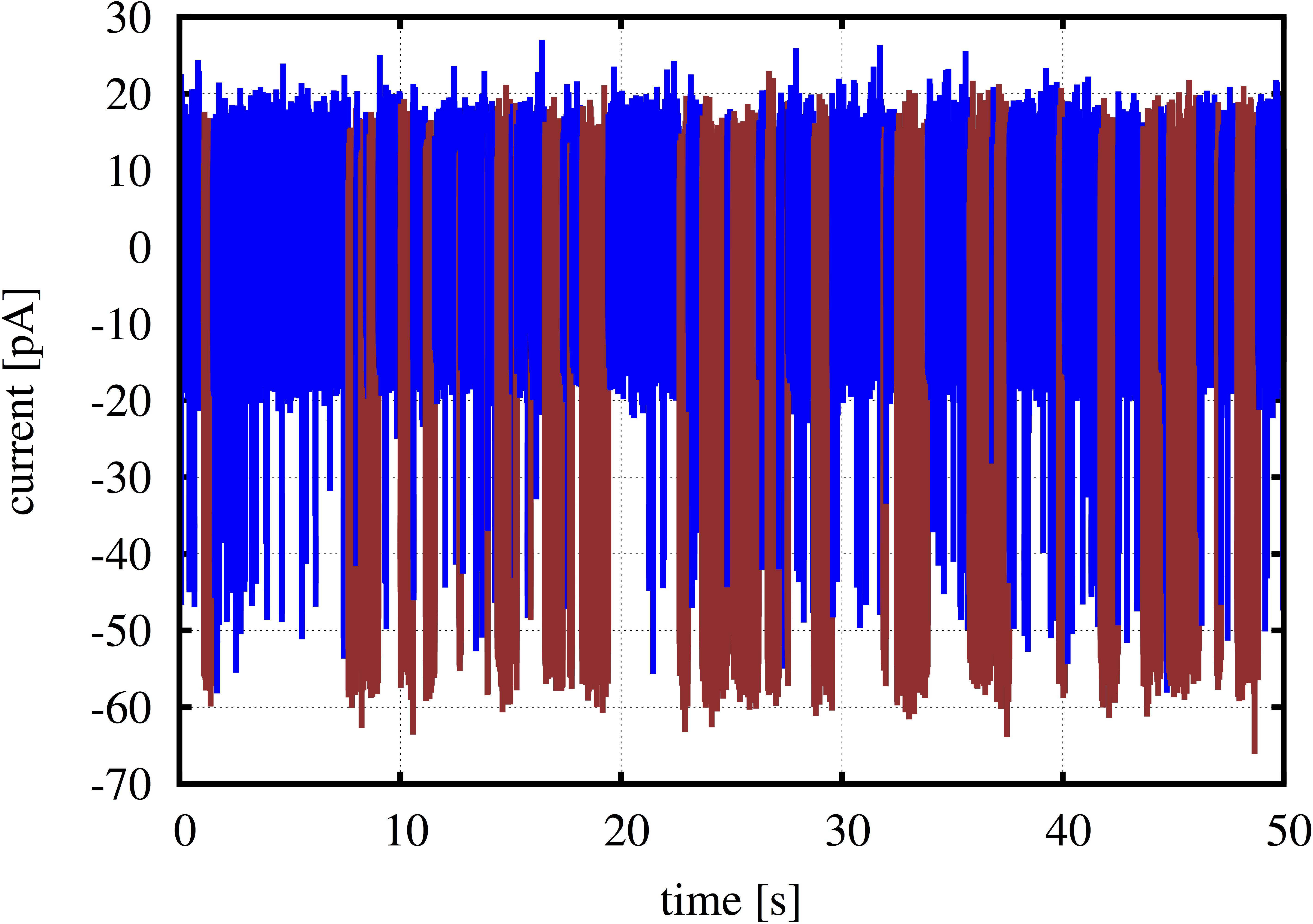}%
  }%
  \caption{The new hierarchical model represents the kinetics more
    accurately than the model from \citet{Sie:12a}. Panel (a) shows
    data from type~2 \ipr\ recorded at \unit{10}{\mu M} \ipthree,
    \unit{5}{\milli M} ATP and \unit{0.05}{\micro M} \ca. The colour
    of the line shows in which mode the channel is for a given point
    in time as inferred by the method of \citet{Sie:14a}. Blue
    indicates the nearly inactive mode~\Mone\ whereas brown indicates
    the active mode~\Mtwo. Similar to the data, the time spent in each
    mode spans a wide range of time scales and the channel alternates
    between the modes relatively infrequently whereas the model from
    \citet{Sie:12a} switches too often.}
  \label{fig:compkinetics}
\end{figure}

\newpage

\section{Mathematical analysis of the hierarchical Markov model}
\label{sec:analysis}

{In the previous section we demonstrated that the hierarchical Markov
  model introduced in Section~\ref{sec:methods} provides a
  statistically efficient framework for systematically building models
  for modal gating. Now, we focus on some interesting aspects of the
  mathematical structure of the hierarchical Markov model and show
  that many important properties of the infinitesimal
  generator~$\fullM$ of the full model can be derived from the
  components~$\hmm$ of the model.}

{In Section~\ref{sec:eigenvalues} we calculate the
  eigenvalues of~$\fullM$. The spectrum of~$\fullM$ consists of two
  parts: the eigenvalues of~$\modesM$ and a subset of the eigenvalues
  of the blocks~$\fullM^{i,i}=\modesM^{i,i} \oplus \Qi$. But whereas
  the eigenvalues of the submatrices~$\modesM^{i,i}$ appear in the
  spectrum of the submatrices~$\fullM^{i,i}$, they are not eigenvalues
  of the full model~$\fullM$.}

{From a modelling point of view it is an important question if
  properties of the components~$\hmm$ are preserved when they are
  combined in the full model. In Section~\ref{sec:modalsojourns} we
  demonstrate that the sojourn time distribution in the states
  representing a particular mode in the model~$\modesM$ is preserved
  for the analogous distribution calculated for the augmented state
  space of~$\fullM$.}

{When the initial distributions~$p^i$ coincide with the
  stationary distributions, $p^i=\statQi$, we calculate the full time-dependent solution and the stationary distribution of~$\fullM$ from the components~$\hmm$ of the hierarchical Markov model (Section~\ref{sec:full}).}




\subsection{Eigenvalues}
\label{sec:eigenvalues}

Before we calculate the eigenvalues for general infinitesimal
generators~$\fullM$ of the full model we remark that in most cases
relevant for ion channel modelling we may assume that the
matrices~$\modesM$ and~$\Qi$ appearing in our model are
diagonalisable---this is implied by the so-called detailed balance
conditions:

\begin{equation}
  \label{eq:detbal}
  \pi^iq_{ij} = \pi^j q_{ji},
\end{equation}
where~$\pi$ is the stationary distribution of an infinitesimal
generator~$Q=(q_{ij})$. A matrix $Q=(q_{ij})$ with~\eqref{eq:detbal}
is diagonalisable with real eigenvalues because by choosing the
transformation matrix $\diag(\statQ)^{1/2}$ it is similar to a
symmetric real matrix. Detailed balance is usually assumed to hold for
ion channel models because it can be related to thermodynamic
reversibility of the transitions between different states in the
model. Note that~\eqref{eq:detbal} holds automatically if the
adjacency graph of the states of a Markov model is acyclic. This
follows from Kolmogorov's criterion \cite{Kol:36a}, see theorem 1.8 of
\citet{Kel:11a} for a more recent statement of the continuous-time
version.  Thus, in particular, all infinitesimal generators~$\modesM$
and~$\Qi$ considered in this article satisfy detailed balance.

{
  \begin{proposition}[Eigenvalues and eigenvectors of~$\fullM$
    assuming detailed balance]
    We assume that the matrices~$\modesM$ and~$\Qi$ of a hierarchical
    Markov model fulfil the detailed balance
    conditions~\eqref{eq:detbal}.
  \begin{enumerate}
  \item Let~$\fullEv$ be an eigenvalue of the matrix~$\modesM$ and
    $v_{\boldsymbol{m}}^T$ a right eigenvector associated with~$\fullEv$.
    Then~$\fullEv$ is also an eigenvalue of the full model~$\fullM$
    with associated right
    eigenvector~$v_{\boldsymbol{m}}^T \otimes u_{\boldsymbol{n}}^T$ where
    $u_{\boldsymbol{n}}^T$ is a vector of~$|\boldsymbol{n}|$ ones.
  \item Moreover, all~$\nu=\modesEv+\lambda$, where~$\modesEv$ is an
    eigenvalue of~$\modesM^{i,i}$ and~$\lambda \neq 0$ is an
    eigenvalue of~$\Qi$, are eigenvalues of the full
    model~$\fullM$. If~$\tilde{w}^i$ is a left eigenvector of the
    submatrix~$\fullM^{i,i}$ associated with the eigenvalue~$\nu$,
    $w_{\boldsymbol{m}}=(0; \dots; 0; \tilde{w}^i; 0; \dots; 0)$
    with~$w(i)=\tilde{w}^i$ and~$w(j)=0$, $i\neq j$ is a left
    eigenvector of~$\fullM$ associated with~$\nu$.
  \end{enumerate}
\end{proposition}

\begin{proof}
  Detailed balance implies that~$\modesM$ and the~$\Qi$ are diagonalisable with real eigenvalues. In particular, all matrices have full sets of eigenvectors. This enables us to construct eigenvectors of the infinitesimal generator~$\fullM$ of the full model from the eigenvectors of~$\modesM$ and the~$\Qi$.

  \begin{enumerate}
  \item We need to show
    that~$M (v_{\boldsymbol{m}}^T \parttens u_{\boldsymbol{n}}^T)=\fullEv
    (v_{\boldsymbol{m}}^T \parttens u_{\boldsymbol{n}}^T)$.
    Let $[M (v_{\boldsymbol{m}}^T \parttens u_{\boldsymbol{n}}^T)]^{i}$ denote
    the~$i$-th component of the partitioned vector. Here,
    $v_{\boldsymbol{m}}^T \parttens u_{\boldsymbol{n}}^T$ is a tensor product
    that is consistent with the partitions~$\boldsymbol{m}$
    and~$\boldsymbol{n}$ as in~\eqref{eq:tensorpartition}
    (Definition~\ref{def:partvector}). We calculate:

    \[
    [M (v_{\boldsymbol{m}}^T \parttens
    u_{\boldsymbol{n}}^T)]^{i}=(\modesM^{i,i}\oplus \Qi)( (v^i)^T \otimes
    u_{n_i}^T) +\sum_{k\neq i} (\modesM^{i,k}\otimes
    P^{i,k})((v^k)^T\otimes u_{n_k}^T).
      \]
Using the compatibility condition of matrix multiplication and tensor product~\eqref{eq:compMatrixTensor} we calculate:

\begin{align*}
  [M (v_{\boldsymbol{m}}^T \parttens
  u_{\boldsymbol{n}}^T)]^{i}&=(\modesM^{i,i} (v^i)^T \otimes u_{n_i}^T + (v^i)^T
                        \otimes \Qi u_{n_i}^T) +\sum_{k\neq i}
                        (\modesM^{i,k} (v^k)^T \otimes P^{i,k}
                        u_{n_k}^T).
\end{align*}
Noting that~$\Qi u_{n_i}^T=0$ and~$P^{i,k} u_{n_k}^T= u_{n_i}^T$ we
finally get:

\begin{align*}
  [M (v_{\boldsymbol{m}}^T \parttens
  u_{\boldsymbol{n}}^T)]^{i}&=\sum_{k=1}^{\nM} \modesM^{i,k} (v^k)^T \otimes
                          u_{n_i}^T=\fullEv ( (v^i)^T \otimes u_{n_i}^T). 
\end{align*}
Because this holds for all blocks we obtain the desired result.

\item All except for the $i$-th block of~$w$ are zero, so we get:

  \begin{align*}
    w \fullM & = (\tilde{w}^i[\modesM^{i,1}
      \otimes P^{1,i}]; \dots; \tilde{w}^i[\modesM^{i,i} \oplus Q^i]; \dots;  \tilde{w}^i [\modesM^{i,\nM}
      \otimes P^{\nM,i}] ).
\end{align*}

Because~$\tilde{w}^i$ is an eigenvector of~$\modesM^{i,i} \oplus Q^i$
we know that~$\tilde{w}^i (\modesM^{i,i} \oplus Q^i)=\nu \tilde{w}^i$.
For~$w$ to be an eigenvector, it remains to be shown that all other
blocks vanish. Let $u$ be a left eigenvector
of~$\modesM^{i,i}$ 
associated with the eigenvalue~$\modesEv$ and $v$ a left eigenvector
of~$Q^i$ associated with the eigenvalue~$\lambda$. Then~$\tilde{w}^i$
can be written as $\tilde{w}^i=u \otimes v$ according
to~\eqref{eq:ev}. Substituting this
and~$P^{i,k}=p^k\otimes u_{n_i}^T$, $k\neq i$, we calculate:

\begin{align}
\label{eq:vanish}
  (  u \otimes v)[\modesM^{1,k} \otimes p^k \otimes u_{n_i}^T] &=u (\modesM^{1,k} \otimes p^k) \otimes v u_{n_i}^T. 
\end{align}

The term~$v u_{n_i}^T$ is the standard scalar
product~$\langle v^T, u_{n_i}^T \rangle$ of the vectors~$v^T$
and~$u_{n_i}^T$. Because the row sums of~$\Qi$ are zero, $u_{n_i}^T$
is in the right nullspace of~$\Qi$. By assumption, $v$ is an
eigenvector associated to any eigenvalue~$\lambda\neq 0$. This means
that~$v$ is not in the left nullspace of~$\Qi$, so it must be
orthogonal to any vector in the right nullspace. It follows
that~\eqref{eq:vanish} vanishes as required.
\end{enumerate}
\end{proof}
}

{For the general case where the infinitesimal generators of
  the model~$\modesM$ and the submatrices~$\fullM^{i,i}$ may not
  necessarily be diagonalisable we need the Schur decomposition
  (Proposition~\ref{the:schur}). The Schur decomposition ensures that
  the matrix~$\fullM$ can be transformed to an upper-triangular matrix
  by a unitary matrix. In the following we construct a unitary
  matrix~$S$ from the components~$\hmm$ of our model.

\begin{lemma}[Unitary matrix~$S$]
\label{the:unitary}
For the components $\hmm$ of a hierarchical Markov model, let 

\[
T_{\modesM} =\Theta^* \modesM \Theta, \quad T_{\modesM^{i,i} \oplus \Qi} =(V_i\otimes W_i)^* \modesM^{i,i} \oplus \Qi (V_i\otimes W_i),
\]
be the Schur decompositions of~$\modesM$
and~$\modesM^{i,i} \oplus \Qi$. Let
$\bar{u}_{n_i}^T=1/\sqrt{n_i} u_{n_i}^T$ be the vectors obtained by
normalising the vectors of ones~$u_{n_i}^T$.

\begin{enumerate}
\item The matrices~$W_i$ may be chosen so that they have the
  form~$W_i=\left(\bar{u}_{n_i}^T \vline \tilde{W}_i\right)$
  with~$\tilde{W}_i\in\mathbb{C}^{n_i \times (n_i-1)}$.
\item 
  Let 
\[
\Theta =
\begin{pmatrix}
  \Theta^1\\
  \hline
\vdots\\
  \Theta^{\nM}
\end{pmatrix}
\]
be row-partitioned according to the block structure of~$\modesM$
from~\eqref{eq:Mmodel}. Then the matrix

  \begin{equation}
    \label{eq:S}
    S =
        \begin{pmatrix}
          \Theta^{1} \otimes \bar{u}_{n_1}^T & \vline & V_1 \otimes \tilde{W}_1  & \vline
          &0 & \vline &
          \dots &\vline & 0 \\
          \hline
          \vdots &\vline& 0&\vline & V_2 \otimes \tilde{W}_2 & \vline & \dots &\vline&\vdots\\
          \vdots & 
& \vdots  & 
 &  & 
&
          \ddots & 
& 0\\
          \Theta^{\nM} \otimes \bar{u}_{n_{\nM}} ^T & \vline & 0  & \vline
          &\dots & \vline &
          \dots &\vline & V_{\nM} \otimes \tilde{W}_{\nM}
        \end{pmatrix}
  \end{equation}
is unitary.
\end{enumerate}

\end{lemma}

\begin{proof} \noindent 
  \begin{enumerate}
  \item Because the row sums of~$\Qi$ vanish, the
    vector~$\bar{u}_{n_i}^T$ is a right eigenvector of~$\Qi$
    associated with the eigenvalue zero. Without loss of generality we
    can choose~$\bar{u}_{n_i}^T$ as the first column of~$W_i$.
  \item By construction, all column vectors of~$S$ are
    normalised. Thus, it remains to show that they are also pairwise
    orthogonal. By definition, any two distinct column vectors
    appearing in the same block of~$S$ are orthogonal. It is trivial
    that column vectors from different blocks are orthogonal unless
    one of the two appears in the first block of~$S$. Thus,
    let~$\theta^T$ be a column vector of~$\Theta$
    and~$v_i^T \otimes \tilde{w}_i^T$ be a column vector of
    any~$V_i \otimes W_i$. With the shorthand for tensor products
    consistent with partitions~\eqref{eq:tensorpartition} introduced in
    Definition~\ref{def:partvector}, the scalar
    product~$\langle \cdot, \cdot \rangle$ of the two columns is

  \[
    \langle \theta^T_{\boldsymbol{m}} \parttens \bar{u}^T_{\boldsymbol{n}}, (0; \dots; v_i \otimes \tilde{w}_i; \dots 0)^T \rangle = \langle \theta^T_{\boldsymbol{m}}(i) \otimes \bar{u}^T_{n_i}, v_i^T \otimes \tilde{w}_i^T \rangle
\]    
and due to the zeroes in all except for the~$i$-th block, all other
summands vanish. Noting that~$\langle u, v\rangle=u (v^*)^T=u^Tv^*$
can be interpreted as a special case of matrix multiplication (where
`$*$' denotes component-wise complex conjugation) we can
use~\eqref{eq:compMatrixTensor}:

\[
\langle \theta^T_{\boldsymbol{m}}(i) \parttens \bar{u}^T_{n_i}, v_i^T \otimes \tilde{w}_i^T \rangle =\langle \theta^T_{\boldsymbol{m}}(i), v^T_i \rangle \langle \bar{u}^T_{n_i}, \tilde{w}_i^T \rangle.
\]
But because~$\bar{u}^T_{n_i}$ appeared as a column in the original
unitary matrix~$W_i$, the~$\tilde{w}_i^T$ are all orthogonal
to~$\bar{u}^T_{n_i}$ so that the above scalar product vanishes. Thus,
the matrix~$S$ is unitary.
  \end{enumerate}
\end{proof}
}

\begin{proposition}[Eigenvalues of the full model~$\fullM$]
  \label{the:full}
  Let~$\modesEv$ be an eigenvalue of the
  model~$\modesM$. Then~$\modesEv$ is also an eigenvalue of the full
  model~$\fullM$
  . 
  Moreover, all~$\nu=\modesEv+\lambda$, where~$\modesEv$ is an
  eigenvalue of~$\modesM^{i,i}$ and~$\lambda \neq 0$ is an eigenvalue
  of~$\Qi$, are eigenvalues of the full model~$\fullM$.
\end{proposition}

{
\begin{proof}
  We demonstrate that with the matrix~$S$ from~\eqref{eq:S} we obtain
  a Schur decomposition of the matrix~$\fullM$. We need to show
  that~$A=S^* \fullM S$ is upper triangular.
  The block structure of~$S$ is rectangular with
  $\nM \times (\nM + 1)$ blocks which means that~$S^*$ has
  an~$(\nM + 1) \times \nM$ block structure. Thus, the resulting
  matrix~$A$ will have $(\nM + 1) \times (\nM + 1)$ blocks and its
  diagonal will consist of the eigenvalues of~$\modesM$ in the upper
  left block followed by the remaining eigenvalues from the
  submatrices~$\modesM^{i,i}$. We show that all blocks~$A^{i,j}$ are
  upper triangular which implies that~$A$ is indeed upper
  triangular. First, a lengthy calculation shows that~$A^{1,1}$ is a
  block-wise expanded form of $\Theta^* \modesM \Theta$ and thus
  upper-triangular. One can see directly that the remaining elements
  on the block diagonal are
  
  \[
  A^{i,i} = (V_i \otimes \tilde{W}_i)^* (\modesM^{i,i} \oplus \Qi) (V_i \otimes \tilde{W}_i)
  \]
  and therefore all upper triangular.  

  It remains to show that the lower diagonal blocks $A^{i,j}$
  with~$i>j$ vanish.  We will demonstrate that the~$A^{i,j}$ vanish
  provided that

\begin{equation}
\label{eq:orthWi}
\tilde{W}_i^* \bar{u}^T_{n_i}=0.
\end{equation}
Equation~\eqref{eq:orthWi} is just another way of saying
that~$\bar{u}_{n_i}^T$ is orthogonal to all columns
of~$\tilde{W}_i$. But this is true because from
Lemma~\ref{the:unitary}(i) we know that~$\bar{u}_{n_i}^T$ is the first
column of~$W_i$, so it must be orthogonal to all column vectors
of~$\tilde{W}_i$.

We now calculate the subdiagonal blocks~$A^{i,j}$, $i>j$. First we
calculate the blocks~$A^{\cdot,1}$ on the first block column. We
observe that

\[
(M \cdot S)^{k,1}=(\modesM^{k,k}\oplus Q^k)(\Theta^k \otimes \bar{u}_{n_k}^T) + \sum_{j\neq k} (\modesM^{k,j} \otimes P^{k,j})(\Theta^j \otimes \bar{u}_{n_j}^T).
\]
Because~$S^*$ is block-diagonal below the first row we can calculate 

\begin{align*}
A^{k+1,1}=(S^*\cdot M \cdot S)^{k+1,1} &= (V_k\otimes \tilde{W}_k)^*(\modesM^{k,k}\oplus Q^k)(\Theta^k \otimes \bar{u}_{n_k}^T)\\
& + \sum_{j\neq k} (V_k\otimes \tilde{W}_k)^*(\modesM^{k,j} \otimes P^{k,j})(\Theta^j \otimes \bar{u}_{n_j}^T)
\end{align*}
because in the row $(k+1)$-th row of~$S^*$ for~$k=1, \dots, \nM$ only the~$k$-th block is non-zero. By taking advantage of~\eqref{eq:compMatrixTensor} we obtain

\begin{align*}
  A^{k+1,1}
  &= (V_k\otimes \tilde{W}_k)^*(\modesM^{k,k}\Theta^k \otimes \bar{u}_{n_k}^T + \Theta^k \otimes Q^k\bar{u}_{n_k}^T)\\
& + \sum_{j\neq k} (V_k\otimes \tilde{W}_k)^*(\modesM^{k,j} \Theta^j \otimes P^{k,j} \bar{u}_{n_j}^T)\\
&= (V_k\otimes \tilde{W}_k)^*(\modesM^{k,k}\Theta^k \otimes \bar{u}_{n_k}^T)\\
& + \sum_{j\neq k} (V_k\otimes \tilde{W}_k)^*(\modesM^{k,j} \Theta^j \otimes \bar{u}_{n_k}^T)
\end{align*}
where we have used~$Q^k\bar{u}_{n_k}^T=0$ and~$P^{k,j} \bar{u}_{n_j}^T=\bar{u}_{n_k}^T$. Again using~\eqref{eq:compMatrixTensor} we calculate

\begin{align*}
  A^{k+1,1}
  &= V_k^*\modesM^{k,k}\Theta^k \otimes \tilde{W}_k^* \bar{u}_{n_k}^T)\\
& + \sum_{j\neq k} V_k^*\modesM^{k,j} \Theta^j \otimes \tilde{W}_k^* \bar{u}_{n_k}^T)
\end{align*}
This vanishes due to~\eqref{eq:orthWi} as explained above.

For the remaining blocks~$A^{k+1,l+1}$, $k>l=1, \dots, \nM-1$ we
simply calculate

\begin{align*}
A^{k+1,l+1} &= (V_k\otimes \tilde{W}_k)^*(\modesM^{k,l} \otimes P^{k,l})(V_l \otimes \tilde{W}_l)\\
&= (V_k^*\modesM^{k,l} \otimes \tilde{W}_k^* P^{k,l})(V_l \otimes \tilde{W}_l)\\
&= (V_k^*\modesM^{k,l} V_l)  \otimes (\tilde{W}_k^* P^{k,l}\tilde{W}_l).
\end{align*}
Replacing~$P^{k,l}$ by~$\bar{u}_{n_k}^T \otimes p^l$~\eqref{eq:Pij} we
get

\[
A^{k+1,l+1}= (V_k^*\modesM^{k,l} V_l)  \otimes (\tilde{W}_k^* \bar{u}_{n_k}^T) \otimes (p^l\tilde{W}_l)
\]
where---due to the term~$\tilde{W}_k^* \bar{u}_{n_k}^T$---we again
conclude with~\eqref{eq:orthWi} that~$A^{k+1,l+1}$ vanishes.

\end{proof}
}

\subsection{Sojourn times in modes}
\label{sec:modalsojourns}

{We will now investigate the sojourn times within the states that
  represent the modes~$\Mi$. The switching between modes is
  represented by a model with infinitesimal generator~$\modesM$ and
  one can ask if the dynamics is preserved after~$\modesM$ is combined
  with the other components~$\hmm$ to the generator~$\fullM$ of the
  full model . We denote by $f_{\modesM^i}(t)$ the density function of
  the sojourn time in mode~$\Mi$ represented by~$\modesM$ and
  by~$f_{\Mi}(t)
  $ the sojourn time densities of~$\Mi$
  in the augmented state space of the generator~$\fullM$ of the full
  model. If the mode switching dynamics is preserved, the sojourn time
  densities should be equal and we will show that
  indeed~$f_{\Mi}(t)=f_{\modesM^i}(t)$.}

\begin{proposition}[Modal sojourn times]
  For~$f_{\Mi}(t)$, sojourn time densities within mode~$\Mi$ with an initial distribution~$p^0$ as in Definition~\ref{def:initialdistribution}, we
  have~$f_{\Mi}(t)=f_{\modesM^i}(t)$.
\end{proposition}

\begin{proof}
  For simplicity we only treat the case of two aggregates of states,
  $\fullM^1$ and~$\fullM^2$. For the sojourn time within~$\fullM^1$ we
  have

\[
f_{\fullM^1}(t) = p^0 \fullM^{2,1} \exp \left( \fullM^{1,1} t \right) \fullM^{1,2} u_{m_2 n_2}^T
\]
where~$p^0=p_{\modesM^2}^0 \otimes p_{Q^2}^0$ is a suitably normalised
initial state probability distribution. Substituting
from~\eqref{eq:fullM} we obtain for

\begin{align*}
  \exp \left( \fullM^{1,1} t \right) \fullM^{1,2}  &= \exp \left(
                                                     \left[
                                                     \modesM^{1,1}
                                                     \oplus Q^1
                                                     \right] t
                                                     \right)
                                                     \fullM^{1,2} \\
                                                   &=  \left[ \exp
                                                     \left(
                                                     \modesM^{1,1} t
                                                     \right) \otimes
                                                     \exp \left(  Q^1
                                                     t \right) \right]
                                                     \left(
                                                     \modesM^{1,2}
                                                     \otimes P^{1,2}
                                                     \right)
\end{align*}
{where we have used~\eqref{eq:expkronecker} for calculating
the matrix exponential.} Now,

\begin{align*}
                                                   \left[ \exp
                                                     \left(
                                                     \modesM^{1,1} t
                                                     \right) \otimes
                                                     \exp \left(  Q^1
                                                     t \right) \right]
                                                     \left(
                                                     \modesM^{1,2}
                                                     \otimes P^{1,2}
                                                     \right) &=   \left[ \exp \left( \modesM^{1,1} t \right) \modesM^{1,2} \right]  \otimes P^{1,2}
\end{align*}
{according to the compatibility of tensor and matrix
  product~\eqref{eq:compMatrixTensor} which will be used repeatedly
  below. Also note that
  $\exp \left( Q^1 t \right) \cdot P^{1,2} = P^{1,2}$.} Multiplying
this on the right by~$u_{m_2 n_2}^T=u_{m_2}^T \otimes u_{n_2}^T$ leads
to

\[
\left\{ \left[ \exp \left( \modesM^{1,1} t \right) \modesM^{1,2}
  \right] \otimes P^{1,2} \right\} \left( u_{m_2}^T \otimes u_{n_2}^T
\right) = \left[ \exp \left( \modesM^{1,1} t \right) \modesM^{1,2}
  u_{m_2}^T \right] \otimes u_{n_1}^T
\]
{where we have evaluated~$P^{1,2}u^T_{n_2} = u^T_{n_1}$ in the
  right-most term. Analogous calculations will be carried out
  automatically below.}  The above result is now multiplied on the
left by~$\fullM^{2,1} =\modesM^{2,1} \otimes P^{2,1}$:

\begin{align*}
& \left( \modesM^{2,1} \otimes P^{2,1}\right) \left[ \exp \left( \modesM^{1,1} t \right) \modesM^{1,2} u_{m_2}^T \right] \otimes u_{n_1}^T 
= \left[ \modesM^{2,1} \exp \left( \modesM^{1,1} t \right) \modesM^{1,2} u_{m_2} \right] \otimes  u_{n_2}^T.
\end{align*}
Finally we multiply the preceding result on the left
by~$p^0 = p_{\modesM^2}^0 \otimes p_{Q^2}^0$ and compute

\begin{align*}
f_{\fullM^1}(t) 
&=\left( p_{\modesM^2}^0 \otimes p_{Q^2}^0 \right) \left[ \modesM^{2,1} \exp \left( \modesM^{1,1} t \right) \modesM^{1,2} u_{m_2} \right] \otimes  u_{n_2}^T\\
&= \left[ p_{\modesM^2}^0\modesM^{2,1} \exp \left( \modesM^{1,1} t
  \right) \modesM^{1,2} u_{m_2}^T \right] \otimes \left( p_{Q^2}^0
  u_{n_2}^T \right).
\end{align*}
Now, because~$\left( p_{Q^2}^0 u_{n_2}^T \right)=1$ we obtain the
desired result:

\begin{align*}
f_{\fullM^1}(t) &= p_{\modesM^2}^0\modesM^{2,1} \exp \left( \modesM^{1,1} t \right) \modesM^{1,2} u_{m_2} ^T = f_{\modesM^1}(t).
\end{align*}

\end{proof}


\subsection{Full solution for $p^i=\statQ^i$}
\label{sec:full}

If we choose initial conditions~$p^i=\statQ^i$, where the~$\statQi$
are stationary distributions of the models~$\Qi$, the solution of the
full model has a particularly simple form.

\begin{proposition}[Full solution for~$p^i=\statQ^i$]
\label{the:fullsolution}
Let $\modessol$ be the time-dependent solution for the initial
condition~$w_{\boldsymbol{n}}^0$ and~$\statM_{\boldsymbol{n}}$ be the
stationary solution of the infinitesimal generator~$\modesM$ with
their partition ${\boldsymbol{m}}$. 
Let~$\statQi$, $i=1, \dots, \nM$ be the stationary distributions
of~$\Qi$ or written as a partitioned vector, $\statQsol$ with its
partition~${\boldsymbol{n}}$. If for each generator~$\Qi$ we
set~$p^i=\statQi$ and we choose an initial
distribution~$\pnull=v_{\boldsymbol{m}}^0 \parttens \statQsol$
consistent with Definition~\ref{def:initialdistribution}, the
solution~$p_{{\boldsymbol{m}}\cdot{\boldsymbol{n}}}(t)$ of the full
model is

\begin{equation}
  \label{eq:fulltime}
  p_{{\boldsymbol{m}}\cdot{\boldsymbol{n}}}(t)=\modessol \parttens \statQsol = (v^1(t)\otimes \statQ^1; \dots; v^i(t) \otimes \statQi; \dots; v^{\nM} \otimes \statQ^{\nM} ).
\end{equation}
By taking the limit~$t\to\infty$ we obtain the stationary distribution
  \begin{equation}
    \label{eq:stationary}
    \statFullM_{{\boldsymbol{m}}\cdot{\boldsymbol{n}}}=\statM_{\boldsymbol{m}} \parttens \statQsol=(\statM^1 \otimes \statQ^1;  \dots; \statMi \otimes \statQi; \dots;
    \statM^{\nM} \otimes \statQ^{\nM}).
  \end{equation}
\end{proposition}

\begin{remark}
  The stationary distribution~\eqref{eq:stationary} is independent of
  the initial distribution~$\pnull$, so, for~$p^i=\statQi$, we
  converge to the stationary distribution~\eqref{eq:stationary} also
  for~$\pnull=(v_{\boldsymbol{m}}^0 \parttens w_{\boldsymbol{n}}^0)$
  with $w_{\boldsymbol{n}}^0\neq \statQsol$ and even for arbitrary
  initial conditions~$\pnull$ that are inconsistent with
  Definition~\ref{def:initialdistribution}.
\end{remark}

\begin{proof}
  That~\eqref{eq:fulltime} is a solution can be shown by substituting~$\fullsol=\modessol \parttens \statQsol$ into 

\begin{equation}
\label{eq:infgen}
\frac{d p(t)}{d t} = p(t) M,
\end{equation}
where~$M$ is the generator of the full model~\eqref{eq:fullM}. First we calculate the left-hand side:

\begin{align}
  \frac{d \fullsol}{d t} &= \frac{d \left( \modessol \parttens \statQsol \right)}{dt} \nonumber\\
&=\left(\frac{d \modessol}{dt}\right) \parttens \statQsol \nonumber\\
&= \left(\modessol \modesM \right) \parttens \statQsol \label{eq:laststepfull}
\end{align}
where the last equality~\eqref{eq:laststepfull} follows
because~$\modessol$ is a solution of the model generated
by~$\modesM$.

We now show that we also obtain~\eqref{eq:laststepfull} from the right-hand side of~\eqref{eq:infgen}. For the~$i$-th component $[\fullsol \cdot M]^i$ we calculate

\begin{align*}
   [\fullsol \cdot M]^i & =\left(v^i(t) \otimes \statQ^i \right)\left(\modesM^{i,i}\oplus \Qi \right) + \sum_{j\neq i} \left( v^j(t) \otimes \statQ^j \right)\left(\modesM^{j,i} \otimes P^{j,i} \right)
\end{align*}
For the first summand the contribution of~$\Qi$ vanishes because
of~$\statQi \Qi=0$
\begin{equation}
\label{eq:solfirst}
\left(v^i(t) \otimes \statQi \right)\left(\modesM^{i,i}\oplus \Qi \right) = \left(v^i(t) \modesM^{i,i} \right) \otimes \statQi + v^i(t) \otimes \statQi \Qi 
=\left(v^i(t) \modesM^{i,i} \right) \otimes \statQi.
\end{equation}
Because of~$\statQ^j P^{j,i}=\statQi$ the second summand simplifies to
\begin{equation}
\label{eq:solsecond}
\sum_{j\neq i} \left( v^j(t) \otimes \statQ^j \right)\left(\modesM^{j,i} \otimes P^{j,i} \right) = \sum_{j\neq i} \left( v^j(t)\modesM^{j,i} \right) \otimes \statQi.
\end{equation}
With~\eqref{eq:solfirst} and~\eqref{eq:solsecond} we derive for each component:

\[
 [\fullsol \cdot M]^i=\sum_{i=1}^{\nM} \left( v^j(t)\modesM^{j,i} \right) \otimes \statQi.
\]
This means that the right-hand side of~\eqref{eq:infgen} is indeed of
the form~\eqref{eq:laststepfull} which confirms
that~\eqref{eq:fulltime} is a solution.

\end{proof}

\newpage

\section{Conclusion}
\label{sec:conclusions}


We have proposed a new model for representing modal gating, the
spontaneous switching of ion channels between different levels of
activity. The model is suitable for modelling channels with an
arbitrary number of modes and is capable of representing both the
probabilistic opening and closing within modes as well as the
stochastic switching between modes that regulates these dynamics.

\subsection{Modular representation of modal gating}
\label{sec:betterstats}

In comparison with previous studies, the model presented here
incorporates modal gating in a much more transparent
way. \citet{Ull:12a} developed their model of the \ipr\ from a binding
scheme. First, the authors determined the set of open and closed model
states from a statistical model selection criterion. Second, they
determined which of these states should account for which of the three
modes observed by \citet{Ion:07a}. The decision that a particular open
or closed state should account for the mode showing a low,
intermediate or high level of activity was based on heuristic
inspection of the ligand-dependency of modal gating. The model was
parameterised by optimising a likelihood that accounted for various
sources of single channel data including statistics of modal
gating. This treats the parameter space of their model as a black box
from which a suitable set of parameters capable of accounting for all
data sets is selected by optimisation. We expect such an approach to
be statistically less efficient than a model whose structure
incorporates modal gating more explicitly.

\citet{Sie:12a} used modal gating as the underlying construction
principle of their model by separating the inference of parameters
related to dynamics within modes from estimation of parameters related
to switching between modes. First, models for the inactive mode~\Mone\
and the active mode~\Mtwo\ were inferred by fitting segments of data
representative of each of the two modes---in fact, the same models
were re-used in the present study. However, because at that time
rigorous statistical techniques for segmenting ion channel data by
modes were not available, the time scales of the switching between
both modes was inferred by connecting the submodels for~\Mone\
and~\Mtwo\ with a pair of transition rates whose values were then
determined from a fit to complete traces of single channel
data. Similar to \citet{Ull:12a} modal gating was thus incorporated
into the model without explicitly considering its stochastic dynamics
apparent in the data.

The model presented here improves the model from \citet{Sie:12a} by
explicitly modelling modal gating. After the stochastic process of
switching between modes has been extracted from the data using a
statistical method such as \citet{Sie:14a} instead of arbitrarily
introducing transition rates between modes as in our previous study,
we 
instead fit a model~$\modesM$ directly to the stochastic process of
mode switching. This enables us to accurately represent mode
switching, only adding exactly as many parameters as required. In
comparison to our previous model, the new model described here
requires only two additional parameters. Inspection of the sojourn
time histograms show that these two parameters are essential in order
to account for the fact that sojourns in the nearly inactive
mode~\Mone\ exhibit two different time scales which cannot be
represented by a model with less parameters.

It is important to note that none of the components $\hmm$ of our
model are determined by fitting to the sequence of open and closed
events observed in experiments---the models~$\Qi$ are inferred from
segments of the data and the model~$\modesM$ is parametrised from
transitions between the modes~$\Mi$. Thus, the open and closed time
distributions~$f_O(t)$ and~$f_C(t)$, respectively, can be considered a
prediction of our hierarchical model~$\fullM$. That the hierarchical
model~$\fullM$ outperforms our previous model whose transition rates
were inferred from a direct fit to complete traces of open and closed
events indicates that the new approach is a superior representation of
the data.


The modular structure of our hierarchical model which separates the
representation of transitions between modes (inter-modal kinetics)
from the dynamics within modes (intra-modal kinetics) not only
provides a more parsimonious representation than previous models but,
most notably, evidence is accumulating that mode switching is more
important for ion channel function than intra-modal kinetics. This was
recently shown in two studies of the role of \ipr\ in intracellular
calcium dynamics. \citet{Cao:14a} showed that the essential features
of calcium oscillations in airway smooth muscle could be preserved
after iteratively simplifying the model from \citet{Sie:12a} to a
two-state model that only accounted for switching between the two
modes neglecting the kinetics of transitions between multiple open and
closed states within the modes. \citet{Sie:15b} applied similar
reduction techniques to demonstrate that also the stochastic dynamics
of small clusters of \ipr s can be captured by a two-state model
reduced to the dynamics of mode switching. In our new hierarchical
model, inter-modal and intra-modal kinetics are represented separately
so that the model representation with the right level of detail can be
chosen based on the requirements of a specific application.


\subsection{Biophysical implications of modal gating}
\label{sec:biophysics}

Although modal gating has been observed for a long time it has rarely
been accounted for in ion channel models
. The crucial importance of modal gating has only recently been
appreciated among investigators of the \ipr\ channel and it is now
widely recognised in the community \citep{Mak:15a}. Various
independent sources of evidence indicate that modal gating must be
accounted for, both for understanding~\ipr\ function as well as for
gaining insight into biophysical properties of the channel
molecule. As mentioned in the previous section, the role of \ipr\ in
intracellular calcium dynamics is defined by its behaviour on the slow
time scale of transitions between different modes rather than the fast
time scale of opening and closing \citep{Cao:14a,
  Sie:15b}. Previously, \citet{Ion:07a} discovered that the~\ipr\
adjusts its level of activity depending on ligands such as calcium by
regulating the proportion of time that the channel spends in different
modes. This was subsequently confirmed by the statistical analysis by
\citet{Sie:14a}. Whereas these results reveal the major functional
implications of modal gating, a detailed analysis of the potassium
channel KscA, discussed in more detail below, gives insight into how
different modes arise from biophysical constraints of the channel
protein \citep{Cha:07a,Cha:07b,Cha:11a}. More recently,
\citet{Vij:15a} published a similar study in acetylcholine
receptors. Also see the commentary by
\citet{Gen:15a}. 
This suggests that modes form a fixed repertoire of possible
behaviours defined by the molecular properties of the channel. Being
constrained to a few different modes, ion channels overcome these
limitations by switching between modes. 


This interpretation implies that appropriate analysis of modal gating
may enable us to extract information on the transitions between
different biophysical states from single channel data which---apart
from giving an accurate representation of its dynamics---has always
been a strong motivation for modelling ion channels. 
The simplest possible representation of an ion channel is a two-state
Markov model with only one open and one closed state. Because opening
of the channel involves a rearrangement of the three-dimensional
structure of the channel protein, known as a conformational change, it
is clear that these two different model states at the same time
correspond to different biophysical states of the channel
protein. Thus, the transition rates between the open and the closed
state provide not just a descriptive representation of the time scale
of opening and closing but, in fact, may stand for the dynamics of a
biophysical process, the conformational change involved with the
opening of the channel. This ``mechanistic'' interpretation explains
the popularity of this type of
model. 
On the one hand the Markov assumption implies that open and closed
times are exponentially distributed which means that durations of
channel openings and closings both have characteristic time
scales~$\tau_O$ and~$\tau_C$ given by the
parameters 
of the exponential sojourn time distributions~$f_O(t)$
and~$f_C(t)$. However, many ion channels exhibit multiple
characteristic open and closed times that cannot be represented by
exponential distributions. On the other hand whereas an open ion
channel must be in a different conformation than a closed ion channel
distinguishing only two conformational states is a very coarse
description of the complicated deformations of channel proteins that
can be identified by molecular dynamics models. Nevertheless, if our
goal is to base our models on rigorous statistical analysis, for some
data we may not be able to identify more than two states.

Non-exponential open and closed times can often be represented
satisfactorily by 
%
aggregated continuous-time Markov models where more than one state is
used for representing the channel being open or closed. These models
provide a simple generalisation of the two-state Markov model and
account for more than just one characteristic open or closed time
scale $\tau_O$ and $\tau_C$. By definition, the sojourn times in the
open or closed class of an aggregated Markov model are distributed
according to a phase-type distribution, a class of distributions
representing the time a Markov chain spends in a set of transient
states until exiting to an absorbing state \citep{Neu:75a, 
  Neu:81a}. 
As with the two-state model it is tempting to also associate the
individual states of an aggregated Markov model with different
biophysical states of the channel
protein. 
The multiple open and closed states of an aggregated Markov model
could be interpreted to resolve in more detail the series of
conformational changes that the channel goes through while it
opens. If this interpretation was valid one could hope to discover
details of the molecular structure of ion channels beyond the trivial
distinction between an open and a closed state once the ``best''
aggregated Markov model for a given data set has been found.

Unfortunately, this ``mechanistic'' interpretation of aggregated
Markov models has several flaws. First, the only reason that a
particular model consists of multiple open and closed states is that
multiple characteristic open and closed times were
observed. Identifying each of these states with a distinct
conformational state relies mostly on the analogy with the two-state
model with at best little and usually no empirical evidence. Neither
experimental techniques nor biophysical modelling approaches currently
available enable us to identify a three-dimensional configuration of
the channel protein that corresponds to a model state with a short
open time and distinguish it from another conformational state that is
characterised by a long open time. If we allow the time scale of
conformational changes to be non-exponentially distributed in general,
multiple open or closed states may actually be associated with the
same conformation. In contrast, it is likely that some conformational
states may not have a strong enough influence on the dynamics that
they are represented by a state in a model inferred from the data. 
Second, and more importantly, aggregated Markov models are
only defined up to equivalence \citep{Fre:85a,Fre:86a,Kie:89a,
  Bru:05a, Sie:12b} with other models having the same number of open
and closed states. In particular, it can be shown that models with
completely different adjacency matrices can describe the same process
\citep{Kie:89a} although there is a canonical phase-type description,
given, for example, by its Laplace-Stieltjes transform. Thus,
interpreting the graphical structure of an aggregated Markov model as
a description of possible transitions between different conformational
states is not necessarily meaningful without further data. A related
problem is the fact that some adjacency matrices lead to
non-identifiable models, in particular, certain types of cyclic models
are non-identifiable. Whereas it is unlikely that transitions between
conformational states underlie any fundamental restrictions of this
kind, only some of these transitions would be identifiable from
experimental data. 
It is important to note that the described challenge of relating
aggregated Markov models with biophysical processes does not restrict
in any way their capability of statistically capturing the stochastic
dynamics of ion channels. This only demonstrates that aggregated
Markov models are a more abstract representation than they may appear
to be at first glance.

In contrast, interpreting mode switching as transitions between
distinct biophysical states does not suffer from these
difficulties. \citet{Cha:07a,Cha:07b,Cha:11a} were able to restrict
the KscA channel to one of its normally four modes by mutating a
particular site of the amino acid sequence of the channel
protein. Combining crystallography imaging and molecular dynamics
modelling they could further demonstrate that the four modes were
related to different conformational states of the channel. It is
therefore likely that switching between distinct characteristic
dynamical patterns in single channel data can be directly associated
with the transition from one to another conformation of the channel
protein. This implies that models which accurately represent mode
switching can also be used to infer the time scales of transitions
between biophysical states associated with these modes. This opens up
the exciting possibility that we can gain insight into biophysical
processes involved in ion channel gating by statistical analysis and
modelling of single channel data rather than having to rely on more
time-consuming experimental techniques such as crystallography or more
laborious modelling techniques such as molecular dynamics.

\section*{Funding}

{This research was in part conducted and funded by the Australian Research Council Centre of Excellence in Convergent Bio-Nano Science and Technology (project number CE140100036). P. Taylor is supported by the Australian Research Council (ARC) Laureate Fellowship FL130100039 and the ARC Centre of Excellence for Mathematical and Statistical Frontiers (ACEMS).}

\newpage
\appendix
{
\section{Mathematical background}
\label{sec:mathsbackground}

The results presented in the main text are derived from the following
properties of the Kronecker product and sum and some well-known
results from linear algebra.

\begin{proposition}[Properties of Kronecker product $\otimes$ and
  Kronecker sum $\oplus$]
  The following properties of the Kronecker product and sums can all
  be found in \citet{Hor:94a}.
  \begin{enumerate}
  \item Transposition and conjugate transpose (Properties 4.2.4 and
    4.2.5
    ):
    \begin{equation}
      \label{eq:transposition}
      (A \otimes B)^T=A^T \otimes B^T, \quad       (A \otimes B)^*=A^* \otimes B^*.
    \end{equation}
  \item Compatibility of tensor product and matrix multiplication
    (Lemma 4.2.10
    ): Let~$A \in \mathbb{R}^{k_1\times m_1}$,
    $C \in \mathbb{R}^{m_1\times n_1}$,
    $B \in \mathbb{R}^{k_2\times m_2}$,
    $D \in \mathbb{R}^{m_2\times n_2}$.
  \begin{equation}
    \label{eq:compMatrixTensor}
    (A\otimes B)(C \otimes D) = (A C) \otimes (B D) \in \mathbb{R}^{k_1 k_2 \times n_1 n_2}.
  \end{equation}
\item Eigenvalues of Kronecker sums~$A \oplus B$ (Theorem
  4.4.5
  ): Let~$\alpha$, $\beta$ denote eigenvalues of the square
  matrices~$A$ and~$B$. Then the eigenvalues of~$\fullM=A \oplus B$
  are
  \begin{equation}
    \label{eq:ev}
    \gamma=\alpha + \beta.
  \end{equation}

\item Matrix exponentials of Kronecker sums (Chapter 6, Problem
  14
  ): For square matrices~$A\in\mathbb{R}^{m\times m}$
  and~$B\in\mathbb{R}^{n\times n}$:
  \begin{equation}
    \label{eq:expkronecker}
    \exp \left(A \oplus B \right) = \exp (A) \otimes \exp (B) \in \mathbb{R}^{mn
      \times mn}.
  \end{equation}
  \end{enumerate}
\end{proposition}

If we cannot assume that a matrix has a complete set of eigenvectors
so that it may not be diagonalisable we can still triangularise this
matrix over the complex numbers~$\mathbb{C}$. The process of
triangulation can be described by the Schur decomposition:

\begin{proposition}[Schur decomposition]
\label{the:schur}
  For a square matrix~$A\in\mathbb{R}^{m \times m}$ there exists a unitary
  matrix~$\Theta \in \mathbb{C}^{m \times m}$ and an upper triangular matrix~$T$
  such that 

  \begin{equation}
    \label{eq:schur}
    T = \Theta^* A \Theta
  \end{equation}
where~$\Theta^*$ is the conjugate transpose of $\Theta$; \eqref{eq:schur} is known as the Schur decomposition. \\ Let~$A\in\mathbb{R}^{m \times m}$ and~$B\in\mathbb{R}^{n \times n}$ with Schur decompositions

\[
T_A=V^* A V, \quad T_B = W^* B W.
\]
Schur decompositions for the Kronecker product~$A\otimes B$ and the
Kronecker sum~$A \oplus B$ can then be obtained via
\begin{equation}
  \label{eq:schurkronecker}
  T_{A \otimes B}=(V \otimes W)^* A\otimes B (V \otimes W), \quad T_{A \oplus B} = (V \otimes W)^* A\oplus B (V \otimes W).
\end{equation}
\end{proposition}

\begin{proof}
  See \citet{Hor:85a}, theorem 2.3.1. For~\eqref{eq:schurkronecker} we
  refer to the proofs of Theorems~4.2.12 and 4.4.5 in \citet{Hor:94a}.
\end{proof}
}


\bibliographystyle{/Users/merlin/latex/style/myapalike2}
\bibliography{/Users/merlin/Documents/references/refer}

\end{document}